\documentclass[10pt, journal, a4paper, comsoc]{IEEEtran}
\usepackage{cite}

\ifCLASSINFOpdf
  \usepackage[pdftex]{graphicx}
  \DeclareGraphicsExtensions{.pdf,.jpeg,-eps-converted-to.pdf}
\else
  \usepackage[dvips]{graphicx}
\fi
\usepackage{amsmath}
\usepackage{amsthm}
\usepackage{bm}
\usepackage{amssymb}
\usepackage{multirow}

\newtheorem{definition}{Definition}
\newtheorem{lemma}{Lemma}

\newtheorem{theorem}{Theorem}

\newcommand{\diag}{\text{diag}}

\newcommand{\trace}{\text{tr}}
\newcommand{\norm}[1]{\left\lVert#1\right\rVert}

\usepackage{enumitem}

\usepackage[ruled,vlined]{algorithm2e}
\makeatletter
\newcommand{\nosemic}{\renewcommand{\@endalgocfline}{\relax}}
\newcommand{\dosemic}{\renewcommand{\@endalgocfline}{\algocf@endline}}
\makeatother
\usepackage{enumitem}
\usepackage{color}

\usepackage{array}
\allowdisplaybreaks[0]
\makeatletter
\@fpsep\textheight
\makeatother
\usepackage{afterpage}

\begin{document}
%
\title{Fairness and Sum-Rate Maximization via Joint Subcarrier and Power Allocation in Uplink SCMA Transmission}
%
%
%

\newcommand{\arxiv}{}
\newcommand{\modulus}{\text{ mod }}

\author{Joao~V.C.~Evangelista,~\IEEEmembership{IEEE Student Member,}
        Zeeshan~Sattar,~\IEEEmembership{IEEE Student Member,} 		~Georges~Kaddoum,~\IEEEmembership{IEEE Member,} and Anas Chaaban,~\IEEEmembership{IEEE Member}
\ifdefined \arxiv
\thanks{
\textsuperscript{\textcopyright} 2019 IEEE.  Personal use of this material is permitted.  Permission from IEEE must be obtained for all other uses, in any current or future media, including reprinting/republishing this material for advertising or promotional purposes, creating new collective works, for resale or redistribution to servers or lists, or reuse of any copyrighted component of this work in other works.

J.V.C. Evangelista, G. Kaddoum and Z. Sattar  were with the Department
of Electrical Engineering, École de Technologie Supérieure, Montreal,
QC, H3C 1K3 CA, e-mail: \{joao-victor.de-carvalho-evangelista.1, zeeshan.sattar.1\}@ens.etsmtl.ca and georges.kaddoum@etsmtl.ca 

Anas Chaaban is with the School of Engineering, University of British Columbia, Kelowna, BC, V1V 1V7 CA, email: anas.chaaban@ubc.ca} 
\else 
\thanks{J.V.C. Evangelista, G. Kaddoum and Z. Sattar  were with the Department
of Electrical Engineering, École de Technologie Supérieure, Montreal,
QC, H3C 1K3 CA, e-mail: \{joao-victor.de-carvalho-evangelista.1, zeeshan.sattar.1\}@ens.etsmtl.ca and georges.kaddoum@etsmtl.ca 

Anas Chaaban is with the School of Engineering, University of British Columbia, Kelowna, BC, V1V 1V7 CA, email: anas.chaaban@ubc.ca} 
\fi
} 
\maketitle

\begin{abstract}
In this work, we consider a sparse code multiple access uplink system, where $J$ users simultaneously transmit data over $K$ subcarriers, such that $J > K$, with a constraint on the power transmitted by each user. To jointly optimize the subcarrier assignment and the transmitted power per subcarrier, two new iterative algorithms are proposed, the first one aims to maximize the sum-rate (Max-SR) of the network, while the second aims to maximize the fairness (Max-Min). In both cases, the optimization problem is of the mixed-integer nonlinear programming (MINLP) type, with non-convex objective functions, which are generally not tractable. We prove that both joint allocation problems are NP-hard. To address these issues, we employ a variant of the block successive upper-bound minimization (BSUM) \cite{razaviyayn.2013} framework, obtaining polynomial-time approximation algorithms to the original problem. Moreover, we evaluate the algorithms' robustness against outdated channel state information (CSI), present an analysis of the convergence of the algorithms, and a comparison of the sum-rate and Jain's fairness index of the novel algorithms with three other algorithms proposed in the literature. The Max-SR algorithm outperforms the others in the sum-rate sense, while the Max-Min outperforms them in the fairness sense.
\end{abstract}

%
\IEEEpeerreviewmaketitle
\section{Introduction}
\label{sec:intro}
%
%
%
%
\IEEEPARstart{T}{he} fifth generation (5G) of wireless networks is expected to deliver better coverage and a higher capacity to massively connected users. One of the fundamental aspects to achieve this goal is the design of multiple access techniques. 
Orthogonal multiple access (OMA) techniques allocate different users into orthogonal network resources, to minimize the interference between users. For instance, time division multiple access (TDMA), code division multiple access (CDMA) and orthogonal frequency division multiple access (OFDMA) assign orthogonal time slots, codes and subcarriers to users, respectively. However, due to the increasing demand for data communications and the introduction of new data-hungry technologies, such as virtual and augmented reality (VR/AR) and massively deployed internet of things (IoT) devices, a tenfold increase in traffic is expected by 2020 \cite{ericsson}. As the number of orthogonal network resources available is finite, this design paradigm is incompatible with the massive traffic and connectivity requirements of 5G networks. Recently, early information-theoretic works on multi-user communications \cite{cover.1972,ahlswede.1973} have reemerged under the name non-orthogonal multiple access (NOMA) as a potential solution to deal with this requirement. Although NOMA methods are rooted in the information-theoretic literature, the recent interest has been focused on communication-theoretic aspects such as developing efficient NOMA coding and modulation schemes, with desired error-rate performance and multi-user communication capabilities. Differently from OMA, in NOMA techniques, multiple users are allocated to the same network resources, permitting the allocation of more users and more efficient use of the available resources. In NOMA, each receiver must perform multiuser detection (MUD) to recover the intended transmitted signal. NOMA techniques can be classified into two different groups, power division NOMA (PD-NOMA), code division NOMA (CD-NOMA). Recently, power domain sparse code multiple access (PSMA) \cite{moltafet.2018a} has been proposed as a hybrid of PD-NOMA and CD-NOMA. An extensive performance comparison of NOMA methods in a single cell system is found in \cite{wang.2015}, while the comparison between PD-NOMA and CD-NOMA in heterogeneous network is presented in \cite{moltafet.2018a}

In CD-NOMA, the transmitter introduces redundancy to the transmitted symbol, via code and/or spreading, to enable receivers to perform MUD and separate signals from different users. Furthermore, CD-NOMA has additional advantages in comparison to PD-NOMA \cite{moltafet.2018}, such as the coding gain and the shaping gain (i.e., methods using multidimensional constellations) \cite{nikopour.2013}.  
Motivated by these advantages, this paper is focused on one of the promising CD-NOMA techniques, named sparse code multiple access (SCMA) \cite{nikopour.2013}. In SCMA, sparse multidimensional codebooks are assigned to each user, and each user's data layer is sparsely spread throughout the network resources. In comparison to OFDMA, SCMA allows for more users than subcarriers available to be served simultaneously, while reducing the peak average power ratio (PAPR) due to the sparsity of the subcarrier allocation. SCMA was first proposed in \cite{nikopour.2013}, as a multidimensional generalization of the low density spreading code division multiple access (LDS-CDMA) that yielded better results regarding detection error. In \cite{taherzadeh.2014}, a method to design SCMA codebooks based on lattice coding was proposed. In \cite{nikopour.2014}, a downlink SCMA system is considered, and an algorithm for user pairing along with rate adjustment and a detection strategy is proposed for a multiuser SCMA scheme. It is shown that this scheme can achieve robustness to mobility and high data rates. 

\subsection{Related Work}
\label{sec:related}
Regarding resource management and allocation in SCMA networks, an algorithm to maximize the rate of successful accesses on a random access massive machine communications network is suggested in \cite{xue.2016}. In \cite{luo.2017}, a resource allocation and subcarrier pairing scheme combining OFDMA and SCMA for a dual-hop multiuser relay network is proposed. The problem of assigning SCMA subcarriers to maximize sum rate in uplink transmission is formulated as matching game in \cite{di.2016}. A grant-free contention based uplink SCMA scheme was proposed in \cite{au.2014}. In \cite{li.2016}, the capacity of an SCMA cell with a Gaussian input is derived and a joint subcarrier and power allocation algorithm is proposed. In \cite{dabiri.2018}, three algorithms for dynamic subcarrier allocation are presented and their link-level performance is evaluated, one of which takes user fairness in consideration. However, the system-level capacity of these algorithms is not investigated and their fairness is compared in terms of the bit error rate (BER) difference between the best and the worst user. In \cite{cui.2017}, a low complexity bisection-based power allocation algorithm, aiming to maximize the capacity of the SCMA system with a finite alphabet is proposed. A stochastic geometry framework to obtain the system-wide area spectral efficiency of underlaid and overlaid device-to-device (D2D) SCMA networks is developed in \cite{liu.2017} and a power allocation strategy to minimize cross-tier interference in underlaid mode and an optimal subcarrier allocation for the overlaid mode are presented. In \cite{zhu.2017}, a joint subcarrier and power allocation algorithm to maximize the proportional fairness utility function of the downlink SCMA system is proposed. The subcarrier and power allocation are split into two problems. The power allocation problem is transformed to a convex equivalent and the remove-and-reallocate algorithm is proposed to solve the combinatorial subcarrier problem. A similar technique of convexification and alternating optimization is employed in \cite{abedi.2018} to solve a SCMA resource allocation problem taking into account content caching, energy harvesting and physical layer security.

Despite the extensive body of literature regarding the analytical characterization and resource management in SCMA, a few core issues are yet to be properly addressed. Firstly, the network overloading achieved by non-orthogonal scheduling also results in an additional source of interference. Hence, it is fundamental to approach the resource allocation problem from a fairness perspective, as algorithms that maximize the sum-rate do so at the expense of users with poor channel condition. In this vein, one of the algorithms proposed in this paper follows a fairness maximization path. Secondly, as shown in this manuscript, the joint subcarrier and power allocation problem is NP-hard. The algorithms currently proposed in the literature propose heuristics to achieve sub-optimal solutions. In face of that, we propose a more systematic approach by relaxing the problem and following the BSUM framework. 

To the best of our knowledge, no previous works have investigated the fairness in joint subcarrier and power allocation in uplink SCMA transmission in depth. Furthermore, the algorithms proposed in this paper have stronger optimality guarantees in comparison with algorithms proposed in previous works.

\subsection{Contributions}
\label{sec:contrib}
In this paper, we formulate two optimization problems for joint subcarrier and power allocation in SCMA networks, one aiming to maximize the sum-rate and another one for maximizing the fairness and propose two algorithms to solve them. The first algorithm's goal is to maximize the sum-rate of the network. While this is an essential criterion in cellular networks, fairness between users is equally important. Thus, to include fairness in the optimization, we propose the Max-Min algorithm aimed to maximize the minimum rate among the users. The obtained results demonstrate better performance than the former algorithm in terms of fairness, at the cost of a lowered sum rate.

Both problems are of the non-convex mixed integer nonlinear programming (MINLP) type. We prove that both problems are NP-hard \footnote{As shown in \cite{moltafet.2018b, mokdad.2019} for the energy-efficiency and heterogeneous cloud radio access networks in PD-NOMA networks respectively, this problems can be reformulated as a monotonic optimization problem, and the optimal joint allocation can be found using the polyblock outer-approximation algorithm, albeit, the algorithm complexity grows exponentially with the size of the problem.}. Then, we propose two algorithms based on the BSUM framework, proposed in \cite{razaviyayn.2013}. The proposed algorithms maximize a lower bound approximation of the objective functions by updating the optimization variables in blocks. As shown in \cite{razaviyayn.2013}, if the lower bound approximation satifies some conditions, this approach has guaranteed convergence to a stationary point, assuring a locally optimal solution. Additionally, we compare both algorithms and the ones proposed in \cite{dabiri.2018} in the sum-rate and the Jain's fairness index sense. Results show that the Max-SR algorithm outperforms all other algorithms regarding sum-rate, while the Max-Min algorithm outperforms all others regarding fairness. Furthermore, we evaluate the fairness and sum-rate performance of the algorithms under outdated CSI. Finally, we compare the BER performance of the two proposed algorithms. 

To summarize, the list below presents the main accomplishments in this work:
\begin{itemize}
\item We prove that the joint power and subcarrier allocation problem is NP-hard.
\item We propose a Max-SR algorithm which achieves a better sum-rate in comparison to the ones proposed in \cite{dabiri.2018}.
\item We propose a Max-Min algorithm which achieves better fairness, in terms of the Jain's fairness index, in comparison to the ones proposed in \cite{dabiri.2018}.
\item We evaluate the robustness of the algorithms against outdated channel state information (CSI)
\end{itemize}

This paper is organized as follows: 
Section \ref{sec:sys_model} contains a brief overview of the SCMA encoder and decoder structure. Also, a description of SCMA signals, and the derivation of its sum-rate is presented. In Section \ref{sec:analysis}, the optimization problems are formulated, and an algorithm for sum-rate maximization, and, another for fairness maximization are proposed. Furthermore, in Section \ref{sec:results} numerical results are shown, and the performance of the algorithm is evaluated. Also, a numerical analysis of the convergence is presented. Finally, in Section \ref{sec:conclusions} the conclusions are presented.

\subsection{Notation}
\label{sec:notation}
Throughout this paper, italic lowercase letters denote real and complex scalar values, and $x^*$ denotes the complex conjugate of $x$. Lower case boldface letters denote vectors, while upper case boldface denote matrices. A lowercase letter with one subscript, $x_i$, represents the $i$-th element of the vector $\mathbf{x}$, while both $x_{i,j}$ and $[\mathbf{X}]_{i, j}$ are used to denote the element on the $i$-th row and $j$-th column of matrix $\mathbf{X}$. The operators $\mathbf{x}^H$ and $\mathbf{X}^H$ denote the hermitian conjugate of a vector and of a matrix, respectively. The operator $\det(\mathbf{X})$ is the determinant of the square matrix $\mathbf{X}$ and $\trace(\mathbf{X})$ is its trace. The operator $\diag(\mathbf{x})$ denotes a square matrix with its diagonal components given by $\mathbf{x}$. The operator $E(\cdot)$ denotes the expected value of a random variable. The function $p(\cdot)$ represents the probability density function (PDF) of a random variable and $\mathbf{x} \sim \mathcal{CN}(\bm{\mu}, \mathbf{K})$, where $\mathbf{K} \in \mathbb{R}^n$, denotes that $\mathbf{x}$ is a complex Gaussian random vector, with mean $\bm{\mu}$ and covariance matrix $\mathbf{K}$. The sets $\mathbb{R}$, $\mathbb{C}$ and $\mathbb{B}$ are the sets of the real, complex and binary numbers, respectively. A calligraphic uppercase letter, such as $\mathcal{X}$, denotes a set and $|\mathcal{X}|$ is its cardinality. The function $\ln(\cdot)$ denotes the natural logarithm of its argument, while the function $I(\cdot ; \cdot)$ is the mutual information between two random variables.

\section{System Model}
\label{sec:sys_model}

\begin{figure*}[!ht]
\centering{	
  \ifCLASSOPTIONtwocolumn 
  	\includegraphics[width=2\columnwidth]{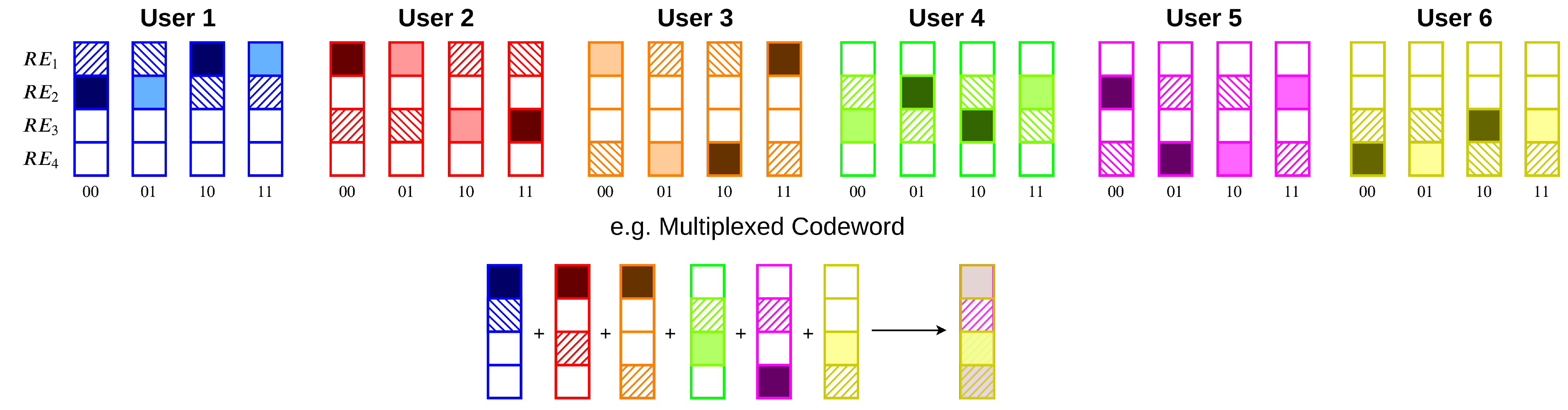}
  \else
  	\includegraphics[width=1\columnwidth]{figs/scma_const-eps-converted-to.pdf}
  \fi
  \caption{Example of an SCMA uplink system with $J = 6$, $K = 4$, $N = 2$ and $d_f = 3$. The square arrays demonstrate the codebook of each user and each square represent the available resource elements (RE). An empty square indicates that no signal is transmitted in the RE and different filling patterns indicate a different complex value.
  		   \label{fig:scma_multiplex}}
}
\end{figure*}

\subsection{SCMA Overview}

\begin{figure}[ht!]
\centering{
  \ifCLASSOPTIONtwocolumn 
		\includegraphics[width=0.5\columnwidth]{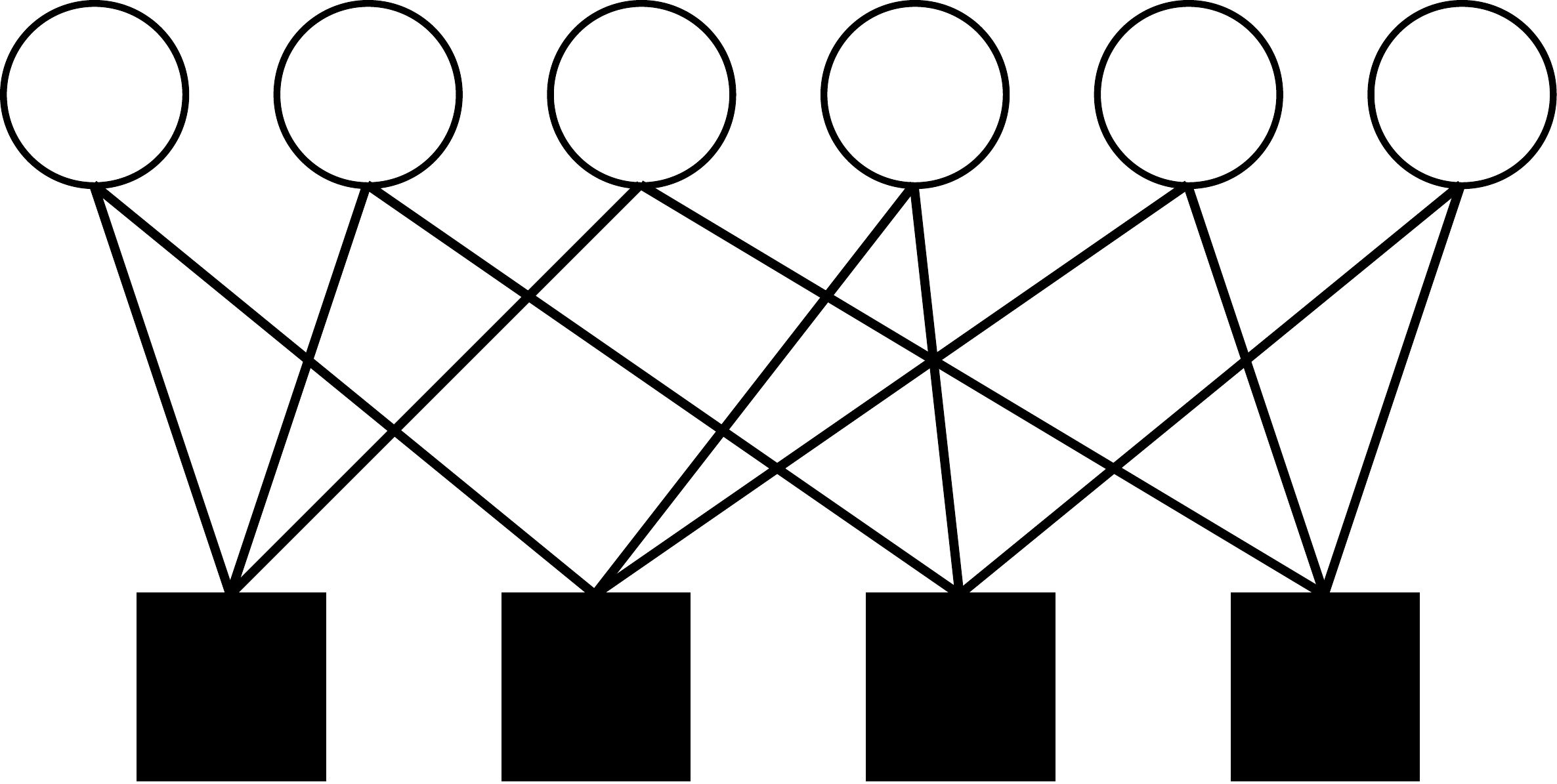}
  \else
    	\includegraphics[width=0.35\columnwidth]{figs/ex_factor_graph-eps-converted-to.pdf}
  \fi
  
  \caption{Example factor graph with $J = 6$, $K = 4$, $N = 2$ and $d_f = 3$. The circles denote user nodes and the squares denote resource nodes. 
  		   \label{fig:factor_graph}}
}
\end{figure}

Consider a system consisting of one base station (BS), and let $\mathcal{K}$ be the set of available resources (OFDMA subcarriers, MIMO spatial layers and so on), with $|\mathcal{K}| = K$, and $\mathcal{J}$ be the set of users served by the BS, with $|\mathcal{J}| = J$. Each user transmits a symbol from a multidimensional constellation with order $M$. The SCMA encoder is a mapping $f: \mathbb{B}^{\log_2(M)} \rightarrow \mathcal{S}_j$, with $\mathbf{s}_j = f(\mathbf{b}_j)$, where $\mathbf{b}_j \in \mathbb{B}^{\log_2(M)}$ is a vector of bits taken at the output of a channel encoder, $\mathcal{S}_j \subset \mathbb{C}^{K}$, $|\mathcal{S}_j| = M$ and $\mathbf{s}_j$ is a sparse vector with $N < K$ nonzero elements for all $j \in \mathcal{J}$. Each user encodes its transmitted signal from a different multidimensional constellation $\mathcal{S}_j$. 
Therefore,  the BS serves up to $J = {{K}\choose{N}}$ users simultaneously and up to $d_f = {{K-1}\choose{N-1}}$ users are allocated on the same resource. The overloading factor of the cell is given by $\lambda = J/K$. Figure \ref{fig:scma_multiplex} shows an example of codebooks and a multiplexed codeword for an SCMA system with $K=4$, $J=6$, $N=2$ and $d_f=3$.  In this figure, each square represents a subcarrier and the different colors represents the codebook of a different user. The texture of the squares is a different symbol from the user's mother constellation, while the blank square indicates that no signal is transmitted at the subcarrier by the user.  In the second row, we give an example of the resulting received signal, which is a superposition of the transmitted symbol by each user, for the transmission of an arbitrary pair of bits by each user.

Optimal SCMA decoding is achieved by maximum a posteriori (MAP) decoding. However, due to the complexity of MAP decoding, message passing algorithms (MPA) that achieve near-optimal decoding, such as belief propagation (BP) \cite{mceliece.1998} are employed, resulting in a complexity of $\mathcal{O}(M^{d_f})$. In order to reduce the decoding complexity of SCMA, alternative receiver architectures have been proposed, such as the SIC-MPA decoder \cite{3gpp,zou.2015} which is a hybrid of the SIC and MPA procedure, and the list spherical decoding (LSD) algorithm \cite{wei.2017}. 

The structure of the SCMA code can be neatly conveyed through a factor graph representation. Let $\mathbf{F} \in \mathbb{B}^{K \times J}$ be the factor graph matrix, each element $f_{k,j}$ indicates if any information from the user $j$ is transmitted on resource $k$. Figure \ref{fig:factor_graph} illustrates a factor graph with $J = 6$, $K = 4$, $N = 2$ and $d_f = 3$ corresponding to the codebook shown in Figure \ref{fig:scma_multiplex}, where the circular vertices represent each user, the squared vertices denote the resources and the edges between them the allocation of a user to specific resources. 
The reader may refer to \cite{nikopour.2013} for more details on the encoder/decoder structure of SCMA. 

In a SCMA system, the signal received by the BS at the resource $k$ can be written as
\begin{equation}
\label{eq:scma_rx_per_res}
y_{k} = \sum_{j \in \mathcal{J}} f_{k,j} h_{k,j} s_{k,j} + n_k,
\end{equation}
where $h_{k, j}$ is the channel coefficient, $s_{k,j}$ is the symbol transmitted from user $j$ on the $k$-th resource, with average power $p_{k,j} = E(|s_{k,j}|^2)$, and  $n_k$ is the $k$-th component of $\mathbf{n} \sim \mathcal{CN}(0, \sigma_n^2 \mathbf{I})$. Here, we assume that $h_{k, j} = \frac{g_{k, j}}{\sqrt{1 + r_j^\alpha}}$. Without loss of generality, we assume $g_{k,j}$ is a Rayleigh distributed random variable representing the small scale fading of the channel of user $j$ on subcarrier $k$, $r_j$ is the distance of user $j$ from the BS and $\alpha$ is the path loss exponent. Throughout this work, it is assumed that the users send a pilot sequence periodically, and, the BS is able to  perfectly estimate the CSI. From (\ref{eq:scma_rx_per_res}), the received signal vector at the BS is written as
\begin{equation}
\mathbf{y} = \mathbf{H} \mathbf{x} + \mathbf{n},
\end{equation}
where $\mathbf{y} \in \mathbb{C}^K$ is a complex vector, $\mathbf{H} \in \mathbb{C}^{K \times KJ}$ is a matrix composed of submatrices, such that, $\mathbf{H} = [\mathbf{H}_1, \mathbf{H}_2, \cdots, \mathbf{H}_J]$, where, $\mathbf{H}_j = \diag([h_{1,j}, h_{2,j}, \cdots h_{K, j}]^T) \text{ } \forall \text{ } j \in \mathcal{J}$. The vector $\mathbf{x} \in \mathbb{C}^{KJ}$ is given by $\mathbf{x} = \begin{bmatrix} \mathbf{x}_1^T & \mathbf{x}_2^T & \cdots & \mathbf{x}_J^T \end{bmatrix}^T$, where $\mathbf{x}_j = \begin{bmatrix} f_{1,j} s_{1, j} & f_{2,j} s_{2, j} & \cdots & f_{K,j} s_{K, j} \end{bmatrix}^T \text{ } \forall \text{ } j \in \mathcal{J}$.

In this paper, we consider a centralized resource allocation architecture, where $K$ users periodically transmit a pilot signal to the BS. We assume the BS obtains perfect CSI, solves the optimization problem described in Section \ref{sec:analysis}, and tells each user which subcarriers, power, and code-rate to use for the next period of time. We assume the channels are quasi-static, so that users can encode at a fixed rate for a period of time. The process is repeated periodically, where the allocations are changed, and the BS tells users to change their transmission accordingly. 

The sum-rate of a SCMA system is defined as the maximum mutual information between the received and transmitted signals. Therefore, assuming channel knowledge at the receiver we have 
\begin{eqnarray}
\label{eq:scma_capacity_der}
R^{\text{sum}}_{\text{SCMA}} &=& \max_{p(\mathbf{x})} I(\mathbf{x} ; \mathbf{y}|\mathbf{H} = \mathbf{H}^\prime) \nonumber \\ 
&=& \max_{p(\mathbf{x})} h(\mathbf{y}|\mathbf{H} = \mathbf{H}^\prime) - h(\mathbf{n})  \nonumber \\
&\overset{(a)}{\leq}& \ln[(\pi e)^K \det(\sigma_n^2 \mathbf{I}_K + 
\mathbf{H}^\prime \mathbf{K}_x {\mathbf{H}^\prime}^H)] - K \ln[\pi e \sigma_n^2] \nonumber \\
&=& \ln\left[\det \left(\mathbf{I}_K + \frac{1}{\sigma_n^2} \mathbf{H}^\prime \mathbf{K}_x {\mathbf{H}^\prime}^H \right) \right].
\end{eqnarray}
In (\ref{eq:scma_capacity_der}), the inequality in $(a)$ follows since a Gaussian input maximizes the entropy of a random vector, under a covariance constraint \cite{cover}. In this paper, we are concerned with maximizing this upper bound in the Max-SR algorithm which is referred henceforth as $C_{\text{SCMA}}$. It is worth noting that for an increase on $d_f$ the distribution of $\mathbf{y}$ approaches a multivariate Gaussian, due to the central limit theorem. 
Furthermore, $\mathbf{K}_x \in \in \mathbb{C}^{KJ \times KJ}$ is the covariance matrix of $\mathbf{x}$ and is given by 
\begin{eqnarray}
\label{eq:cov_codebook}
\mathbf{K}_x =
\begin{bmatrix}
E(\mathbf{x}^{}_1 \mathbf{x}^H_1) & E(\mathbf{x}^{}_1 \mathbf{x}_2^H) & \cdots & E(\mathbf{x}^{}_1 \mathbf{x}_J^H) \\
\vdots & \vdots  & \ddots & \vdots \\
E(\mathbf{x}^{}_J \mathbf{x}_1^H) & E(\mathbf{x}^{}_J \mathbf{x}_2^H) & \cdots & E(\mathbf{x}^{}_J \mathbf{x}_J^H) \\
\end{bmatrix}
,
\end{eqnarray}
where, each $E(\mathbf{x}_i^{} \mathbf{x}_j^H)$ is given by
\begin{equation}
\label{eq:cov_elem}
E(\mathbf{x}_i \mathbf{x}_j^H) = 
\begin{bmatrix}
E\left(x_{1, i}^{} x_{1, j}^* \right) & E\left(x_{1, i}^{} x_{2, j}^* \right) & \cdots & E\left(x_{1, i}^{} x_{K, j}^* \right) \\
\vdots 				 & \vdots				& 	\ddots	 & \vdots \\
E\left(x_{K, i}^{} x_{1, j}^* \right) & E\left(x_{K, i}^{} x_{2, j}^* \right) & \cdots & E\left(x_{K, i}^{} x_{K, j}^* \right)
\end{bmatrix}.
\end{equation}
As concluded in \cite{li.2016}, if each nonzero coordinate of $\mathbf{s}$ is drawn from centrally symmetric constellations, the cross correlation between the multidimensional symbols from different users is equal to zero, hence, $E(\mathbf{x}_i \mathbf{x}_j^H)$ is equal to a $K \times K$ zero matrix for any $i \neq j$. Thus, (\ref{eq:cov_codebook}) is diagonal. For a more generic derivation, not relying on the central symmetry of the constellation the reader may refer to \cite{le.2018, zaidel.2018, shental.2017}.

As $\mathbf{H}_j$ is diagonal for all values of $j$, we have
\begin{eqnarray}
\label{eq:cov_channel_prod}
\mathbf{H} \mathbf{K}_x \mathbf{H}^H = \underset{j=1}{\overset{J}{\sum}} \mathbf{H}_j E\left(\mathbf{x}_j^{} \mathbf{x}_j^H \right) \mathbf{H}_j^H  \hspace*{3cm}\nonumber \\
=
\begin{bmatrix}
\underset{j=1}{\overset{J}{\sum}} h_{1,j} E \left( x_{1,j} x_{1,j}^* \right) h_{1,j}^* & \cdots & 0 \\
\vdots & \ddots & \vdots \\
0 & \cdots & \underset{j=1}{\overset{J}{\sum}} h_{K,j} E \left( x_{K,j} x_{K,j}^* \right) h_{K,j}^* 
\end{bmatrix}
\end{eqnarray}
It is worth noting that such codebook satisfies the design principles established in \cite{taherzadeh.2014}. With that being said, the SCMA sum-rate in (\ref{eq:scma_capacity_der}) can be simplified as shown 
\ifCLASSOPTIONtwocolumn
in (\ref{eq:sumrate}), at the top of the next page.
\begin{figure*}
\begin{eqnarray}
\label{eq:sumrate}
C_{\text{SCMA}} 
&=& \ln \det \left(\mathbf{I}_K + \frac{1}{\sigma_n^2} \mathbf{H} \mathbf{K}_x \mathbf{H}^H \right) 
= \ln \det 
\begin{bmatrix}
1 + \frac{\underset{j \in \mathcal{J}}{\sum} |h_{1, j}|^2 E(x_{1, j}^2)}{\sigma_n^2} & \cdots & 0 \\
\vdots & \ddots & \vdots \\
0 & \cdots & 1 + \frac{\underset{j \in \mathcal{J}}{\sum} |h_{K, j}|^2 E(x_{K, j}^2)}{\sigma_n^2}  
\end{bmatrix} \nonumber \\
&=& \underset{k \in \mathcal{K}}{\sum} \ln \left( 1 + \frac{\underset{j \in \mathcal{J}}{\sum} |h_{k, j}|^2 E(x_{k, j}^2)}{\sigma_n^2} \right) 
= \underset{k \in \mathcal{K}}{\sum} \ln \left( 1 + \frac{\underset{j \in \mathcal{J}}{\sum} |h_{k, j}|^2 f_{k, j} p_{k, j}}{\sigma_n^2} \right).
\end{eqnarray}
\hrulefill
\end{figure*}
\else
below
\begin{eqnarray}
\label{eq:sumrate}
C_{\text{SCMA}} 
&=& \ln \det \left(\mathbf{I}_K + \frac{1}{\sigma_n^2} \mathbf{H} \mathbf{K}_x \mathbf{H}^H \right) 
= \ln \det 
\begin{bmatrix}
1 + \frac{\underset{j \in \mathcal{J}}{\sum} |h_{1, j}|^2 E(x_{1, j}^2)}{\sigma_n^2} & \cdots & 0 \\
\vdots & \ddots & \vdots \\
0 & \cdots & 1 + \frac{\underset{j \in \mathcal{J}}{\sum} |h_{K, j}|^2 E(x_{K, j}^2)}{\sigma_n^2}  
\end{bmatrix} \nonumber \\
&=& \underset{k \in \mathcal{K}}{\sum} \ln \left( 1 + \frac{\underset{j \in \mathcal{J}}{\sum} |h_{k, j}|^2 E(x_{k, j}^2)}{\sigma_n^2} \right) 
= \underset{k \in \mathcal{K}}{\sum} \ln \left( 1 + \frac{\underset{j \in \mathcal{J}}{\sum} |h_{k, j}|^2 f_{k, j} p_{k, j}}{\sigma_n^2} \right).
\end{eqnarray}
\fi
Furthermore, by assuming a decoding order starting from user $J$ to user $1$ and using the logarithm identity ${\log_b (a + c) = \log_b(a) + \log_b(1 + \frac{c}{a})}$, it is possible to obtain the achievable rate of user $j$ on resource $k$, $C_{k, j}$, as 
\ifCLASSOPTIONtwocolumn
shown in \eqref{eq:sumrate_decomposition}, at the top of the next page.
\begin{figure*}
\begin{eqnarray}
\label{eq:sumrate_decomposition}
&& \ln \left( 1 + \frac{\underset{j \in \mathcal{J}}{\sum} |h_{k, j}|^2 f_{k, j} p_{k, j}}{\sigma_n^2} \right) \nonumber 
= \ln \left( \frac{\sigma_n^2 + \underset{j \in \mathcal{J}}{\sum} |h_{k, j}|^2 f_{k, j} p_{k, j}}{\sigma_n^2} \right) \nonumber \\
&=& \ln \left(\frac{\sigma_n^2 + |h_{k,1}|^2 f_{k,1} p_{k,1}}{\sigma_n^2} \right) \nonumber 
+ \ln \left(\frac{\sigma_n^2 + |h_{k,1}|^2 f_{k,1} p_{k,1} + |h_{k,2}|^2 f_{k,2} p_{k,2}}{\sigma_n^2 + |h_{k,1}|^2 f_{k,1} p_{k,1}} \right) \nonumber 
+ \cdots + \ln \left(\frac{\sigma_n^2 + \overset{J-1}{\underset{i=1}{\sum}} |h_{k,i}|^2 f_{k,i} p_{k,i} + |h_{k,J}|^2 f_{k,J} p_{k,J}}{\sigma_n^2 + \overset{J-1}{\underset{i=1}{\sum}} |h_{k,i}|^2 f_{k,i} p_{k,i}} \right) \nonumber \\
&=& \underset{j \in \mathcal{J}}{\sum} \ln \left( 1 + \frac{|h_{k,j}|^2 f_{k,j} p_{k,j}}{\sigma_n^2 + \overset{j-1}{\underset{i = 1}{\sum}}|h_{k,i}|^2 f_{k,i} p_{k,i}}  \right)  
= \underset{j \in \mathcal{J}}{\sum} C_{k,j}.
\end{eqnarray}
\hrulefill
\end{figure*}
\else
\begin{eqnarray}
\label{eq:sumrate_decomposition}
&& \ln \left( 1 + \frac{\underset{j \in \mathcal{J}}{\sum} |h_{k, j}|^2 f_{k, j} p_{k, j}}{\sigma_n^2} \right) \nonumber 
= \ln \left( \frac{\sigma_n^2 + \underset{j \in \mathcal{J}}{\sum} |h_{k, j}|^2 f_{k, j} p_{k, j}}{\sigma_n^2} \right) \nonumber \\
&=& \ln \left(\frac{\sigma_n^2 + |h_{k,1}|^2 f_{k,1} p_{k,1}}{\sigma_n^2} \right) \nonumber 
+ \ln \left(\frac{\sigma_n^2 + |h_{k,1}|^2 f_{k,1} p_{k,1} + |h_{k,2}|^2 f_{k,2} p_{k,2}}{\sigma_n^2 + |h_{k,1}|^2 f_{k,1} p_{k,1}} \right) \nonumber 
+ \cdots  \nonumber \\
&+& \ln \left(\frac{\sigma_n^2 + \overset{J-1}{\underset{i=1}{\sum}} |h_{k,i}|^2 f_{k,i} p_{k,i} + |h_{k,J}|^2 f_{k,J} p_{k,J}}{\sigma_n^2 + \overset{J-1}{\underset{i=1}{\sum}} |h_{k,i}|^2 f_{k,i} p_{k,i}} \right) \nonumber \\
&=& \underset{j \in \mathcal{J}}{\sum} \ln \left( 1 + \frac{|h_{k,j}|^2 f_{k,j} p_{k,j}}{\sigma_n^2 + \overset{j-1}{\underset{i = 1}{\sum}}|h_{k,i}|^2 f_{k,i} p_{k,i}}  \right)  
= \underset{j \in \mathcal{J}}{\sum} C_{k,j}.
\end{eqnarray}
\fi
Therefore, the achievable rate of each user, $C_j$, is given by
\begin{equation}
C_j = \underset{k \in \mathcal{K}}{\sum} C_{k,j} = \underset{k \in \mathcal{K}}{\sum} \ln \left( 1 + \frac{|h_{k,j}|^2 f_{k,j} p_{k,j}}{\sigma_n^2 + \overset{j-1}{\underset{i = 1}{\sum}} |h_{k,i}|^2 f_{k,i} p_{k,i}}  \right).
\end{equation}
In the next section, we formulate the joint subcarrier and power allocation problems and present our proposed method to solve them.
\section{Joint Subcarrier and Power Allocation}
\label{sec:analysis}
we formulate and propose two joint subcarrier and power allocation algorithms to solve two optimization problems: maximizing the sum-rate ($\mathbf{P}_{\textbf{Max-SR}}$) and maximizing the minimum rate of users ($\mathbf{P}_{\textbf{Max-Min}}$). The former can be formulated as
\begingroup
\allowdisplaybreaks[0]
\begin{eqnarray}
&&\mathbf{P}_{\textbf{Max-SR}}: \nonumber \\
\label{eq:obj_sr}
&\underset{{\mathbf{P}, \mathbf{F}}}{\max}& C_{\text{SCMA}} = \underset{k \in \mathcal{K}}{\sum} \ln \left( 1 + \frac{\underset{j \in \mathcal{J}}{\sum} |h_{k, j}|^2 f_{k, j} p_{k, j}}{\sigma_n^2} \right) \\
\label{eq:c_max_res_per_user}
&\textbf{s.t.}& \underset{k \in \mathcal{K}}{\sum} f_{k,j} \leq N \text{ } \forall \text{ } j \in \mathcal{J} \\
\label{eq:c_max_user_per_res}
&& \underset{j \in \mathcal{J}}{\sum} f_{k,j} \leq d_f \text{ } \forall \text{ } k \in \mathcal{K} \\
\label{eq:c_max_pwr}
&& \underset{k \in \mathcal{K}}{\sum} f_{k,j} p_{k,j} \leq P^{(j)}_{\text{max}} \text{ } \forall \text{ } j \in \mathcal{J} \\
\label{eq:c_int_ch}
&& f_{k,j} \in \{0,1\} \text{ } \forall \text{ } k \in \mathcal{K} \text{ and } \forall \text{ } j \in \mathcal{J},
\end{eqnarray}
\endgroup
where $\mathbf{P} \in \mathbb{R}^{K \times J}$ is the matrix of allocated power, (\ref{eq:obj_sr}) is the sum-rate, and (\ref{eq:c_max_res_per_user}) is the constraint on the number of subcarriers allocated per user. The constraint on the number of users per subcarrier is given by (\ref{eq:c_max_user_per_res}), while (\ref{eq:c_max_pwr}) is the constraint on the maximum transmitting power available per user. Finally, (\ref{eq:c_int_ch}) is a binary constraint on the values of $f_{k,j}$. 
Furthermore, the problem $\mathbf{P}_{\textbf{Max-Min}}$ is formulated as
\begin{eqnarray}
&&\mathbf{P}_{\textbf{Max-Min}}: \nonumber \\
\label{eq:obj_mm}
&\underset{{\mathbf{P}, \mathbf{F}}}{\max}&  
\underset{j \in \mathcal{J}}{\min} \underset{k \in \mathcal{K}}{\sum} \ln \left( 1 + \frac{|h_{k,j}|^2 f_{k,j} p_{k,j}}{\sigma_n^2 + \overset{j-1}{\underset{i = 1}{\sum}}|h_{k,i}|^2 f_{k,i} p_{k,i}}  \right)\\
&\textbf{s.t.}& (\ref{eq:c_max_res_per_user}), (\ref{eq:c_max_user_per_res}), (\ref{eq:c_max_pwr}), (\ref{eq:c_int_ch}), \nonumber
\end{eqnarray}
where (\ref{eq:obj_mm}) is the max-min utility function of the rate per user. The objective function of this problem is non-concave and, similarly to $\textbf{P}_{\textbf{Max-SR}}$, also has integer constraints on $\mathbf{F}$. Consequently, we can prove the following statement

\begin{theorem}
\label{th:np_hard}
Both the $\mathbf{P_{\text{Max-SR}}}$ and the $\mathbf{P_{\text{Max-Min}}}$ problems are NP-hard.
\end{theorem}
\begin{proof}
See Appendix \ref{app:np_proof}.
\end{proof}

In order to solve both these problems, we relax the integer constraint on matrix $\mathbf{F}$, given in equation (\ref{eq:c_int_ch}), to a continuous one. Afterwards, we apply the block successive lower bound maximization (BSLM), which is the maximization variant of the approach proposed in \cite{razaviyayn.2013}, which converges to a local minimum of the relaxed problem \cite{razaviyayn.2013}.

For the sake of clarity, we list the conditions for the convergence of the BSLM algorithm. First consider the problem below 
\begin{eqnarray}
&\underset{\mathbf{x}}{\max}& f(\mathbf{x}) \nonumber \\
&\textbf{s.t.}& \mathbf{x} \in \mathcal{X} \nonumber,
\end{eqnarray}
where $f:\mathbb{R}^n \rightarrow \mathbb{R}$ is a non-concave and possibly non-smooth function. Let $\mathcal{X}$ be the cartesian product of $m$ closed convex sets: $\mathcal{X} = \mathcal{X}_1 \times \cdots \times \mathcal{X}_m$, with $\mathcal{X}_i \subseteq \mathbb{R}^{n_i}$ and $\underset{i}{\sum} n_i = n$. Furthermore the optimization variable $\mathbf{x} \in \mathbb{R}^n$ can be decomposed into $m$ vectors $\mathbf{x} = (\mathbf{x}_1, \cdots, \mathbf{x}_m)$, such that $\mathbf{x}_i \in \mathcal{X}_i$. At iteration $t$ of the BSLM algorithm, the blocks of optimization variables are updated cyclically, where for each block the following problem is solved
\begin{eqnarray}
&\underset{\mathbf{x}_i}{\max}& \tilde{f}(\mathbf{x}_i, \mathbf{x}^{(t-1)}) \nonumber \\
&\textbf{s.t.}& \mathbf{x}_i \in \mathcal{X}_i, \nonumber
\end{eqnarray}
where $i = t \modulus m$, and $\mathbf{x}^{(t-1)}$ is the previous value of $\mathbf{x}$. The convergence of the BLSM algorithm is guaranteed if the following conditions hold for $\tilde{f}(\mathbf{x}_i, \mathbf{x}^{(t-1)})$:
\begin{eqnarray}
&\tilde{f}(\mathbf{x}_i, \mathbf{x})& = f(\mathbf{x}), \quad \forall \, \mathbf{x} \in \mathcal{X}, \, \forall \, i \label{eq:ass1}\\
&\tilde{f}(\mathbf{x}_i, \mathbf{y})& \leq f(\mathbf{y}_1, \cdots, \mathbf{y}_{i-1}, \mathbf{x}_i, \mathbf{y}_{i+1}, \cdots, \mathbf{y}_m), \nonumber \\ && \quad \forall \, \mathbf{x}_i \in \mathcal{X}_i, \, \forall \, \mathbf{y} \in \mathcal{X}, \, \forall \, i \label{eq:ass2}\\
&\nabla \tilde{f}(\mathbf{x}_i, \mathbf{x})& = \nabla f(\mathbf{x}) \label{eq:ass3}\\
&\tilde{f}(\mathbf{x}_i, \mathbf{y})& \text{ is continuous in } (\mathbf{x}_i, \mathbf{y}), \quad \forall \, i \label{eq:ass4}
\end{eqnarray} 

In the rest of this section, we propose lower bound convex approximations to the objective functions of both $\mathbf{P_{\text{Max-SR}}}$ and $\mathbf{P_{\text{Max-Min}}}$ satisfying conditions (\ref{eq:ass1})-(\ref{eq:ass4}), and finalize by providing a description of the block update algorithm that converges to locally optimal solutions.

\subsection{Solving $\mathbf{P_{\text{Max-SR}}}$}

To solve this problem, we first relax the integer constraint in (\ref{eq:c_int_ch}), and add a penalty term to the objective function, such that, non-integer solutions to $\mathbf{F}$ are penalized. The $\mathbf{P_{\text{Max-SR}}}$ becomes
\begingroup
\allowdisplaybreaks[0]
\begin{eqnarray}
&&\mathbf{P}_{\textbf{Max-SR}}^\prime: \nonumber \\
\label{eq:obj_sr_pen}
&\underset{{\mathbf{P}, \mathbf{F}}}{\max}& \underset{k \in \mathcal{K}}{\sum} \ln \left( 1 + \frac{\underset{j \in \mathcal{J}}{\sum} |h_{k, j}|^2 f_{k, j} p_{k, j})}{\sigma_n^2} \right) + \gamma \left( \mathbf{F} \right) \\
\label{eq:c_rlx_ch}
&\textbf{s.t.}& 0 \leq f_{k,j} \leq 1 \text{ } \forall \text{ } k \in \mathcal{K} \text{ and } \forall \text{ } j \in \mathcal{J} \\
&& (\ref{eq:c_max_res_per_user}), (\ref{eq:c_max_user_per_res}), (\ref{eq:c_max_pwr})
\end{eqnarray}
\endgroup
where $\gamma \left( \mathbf{F} \right) = \lambda \underset{k \in \mathcal{K}}{\sum} \underset{j \in \mathcal{J}}{\sum} \left( f_{k, j}^2 - f_{k,j} \right)$ is the penalty function \footnote{By selecting moderately high values for $\lambda$ (around $20$), integer solutions are obtained. As a matter of fact, higher values for $\lambda$ results in faster convergence to an integer solution, however, it renders the optimization solver iterations more unstable as it contributes to the ill-conditioning of the problem. On the other hand, smaller $\lambda$ leads to more conservative updates of $\mathbf{F}$ at each BSLM step, resulting in slower convergence, but better optimizers.}. Notice that $\gamma \left( \mathbf{F} \right) < 0$ for all non-integer solutions and $\gamma \left( \mathbf{F} \right) = 0$ for integer ones. This gives incentive for the algorithm to obtain solutions that minimize $\gamma \left( \mathbf{F} \right)$, hence leading to integer solutions of $\mathbf{F}$.

There are two issues that make $\mathbf{P}_{\textbf{Max-SR}}^\prime$ a hard problem to solve:
\begin{itemize}
\item The presence of multi-linear terms of the form $f_{k,j} p_{k,j}$ in (\ref{eq:obj_sr_pen}) and (\ref{eq:c_max_pwr}).
\item Even if $\mathbf{F}$ and $\mathbf{P}$  are updated cyclically, the objective function in (\ref{eq:obj_sr_pen}) is non-concave, due to the addition of $\gamma \left( \mathbf{F} \right)$, which is convex.
\end{itemize}
Now we present two Lemmas that are instrumental to the algorithm that finds a locally optimal solution to $\mathbf{P}_{\textbf{Max-SR}}^\prime$ in polynomial time.
\begin{lemma}
\label{lemma:maxsr_cycle}
If $\mathbf{F}$ and $\mathbf{P}$ are updated cyclically in $\mathbf{P}_{\textbf{Max-SR}}^\prime$, the feasible set of the problem solved in each update step is convex.
\end{lemma}
\begin{proof}
All constraints in $\mathbf{P}_{\textbf{Max-SR}}^\prime$ are linear functions of $\mathbf{F}$ and $\mathbf{P}$, with the exception of constraint (\ref{eq:c_max_pwr}) which involves a multi-linear term. However, if $\mathbf{F}$ and $\mathbf{P}$ are updated cyclically, only one of the matrices is updated while the other is kept constant and the multi-linear terms in (\ref{eq:c_max_pwr}) become linear in the variable being updated. Therefore, the feasible sets are convex.
\end{proof}
\begin{lemma}
\label{lemma:maxsr_approx}
Let the function
\begin{eqnarray}
\label{eq:approx_obj_sr}
&& \underset{k \in \mathcal{K}}{\sum} \ln \left( 1 + \frac{\underset{j \in \mathcal{J}}{\sum} |h_{k, j}|^2 f_{k, j} p_{k, j}}{\sigma_n^2} \right) + \nonumber\\
&&\gamma \left(\mathbf{F}^{\prime} \right) + \trace \left[ \nabla \gamma\left(\mathbf{F}^{\prime} \right)^T \left(\mathbf{F} - \mathbf{F}^{\prime} \right)\right],
\end{eqnarray}
where $\nabla \gamma\left(\mathbf{F}^{\prime} \right) \in \mathbb{R}^{K \times J}$ is a matrix such that
\begin{eqnarray}
\left[ \nabla \gamma\left(\mathbf{F^{\prime}} \right) \right]_{k,j} = \left. \frac{\partial \gamma\left(\mathbf{F} \right)}{\partial f_{k,j}} \right \vert_{\mathbf{F} = \mathbf{F^\prime}},
\end{eqnarray}
be an approximation to (\ref{eq:obj_sr_pen}) in the neighborhood of $\mathbf{F}^\prime$ for fixed $\mathbf{P}$.
%
Notice that (\ref{eq:approx_obj_sr}) is a lower bound concave approximation to (\ref{eq:obj_sr_pen}) satisfying conditions (\ref{eq:ass1})-(\ref{eq:ass4}). 
\end{lemma}
\begin{proof}
Firstly, notice that $\gamma \left(\mathbf{F}^{\prime} \right) + \trace \left[ \nabla \gamma\left(\mathbf{F}^{\prime} \right)^T \left(\mathbf{F} - \mathbf{F}^{\prime} \right)\right]$ is the first order linear approximation of $\gamma \left( \mathbf{F} \right)$ in the neighborhood of $\mathbf{F}^{\prime}$. %
So (\ref{eq:ass1}), (\ref{eq:ass3}), and (\ref{eq:ass4}) are satisfied. Furthermore, as $\gamma \left( \mathbf{F} \right)$ is a convex function, we have
\begin{eqnarray}
\gamma(\mathbf{F}) \geq \gamma \left(\mathbf{F}^{\prime} \right) + \trace \left[ \nabla \gamma\left(\mathbf{F}^{\prime} \right)^T \left(\mathbf{F} - \mathbf{F}^{\prime} \right)\right]. \nonumber
\end{eqnarray}
As the linear approximation is globally less than $\gamma \left( \mathbf{F} \right)$, we have that (\ref{eq:approx_obj_sr}) is a lower bound of (\ref{eq:obj_sr}). Thus, (\ref{eq:ass1}) is also satisfied.
\end{proof}
With Lemmas \ref{lemma:maxsr_cycle} and \ref{lemma:maxsr_approx} in hand, we can derive the convergence of the relaxed problem $\mathbf{P}_{\textbf{Max-SR}}^\prime$ to a local optimum, as stated in the Theorem below.
\begin{theorem}
\label{th:max_sr}
By updating $\mathbf{F}$ and $\mathbf{P}$ cyclically with the solutions to $\mathbf{P}_{\textbf{Max-SR}}^{(\mathbf{F})}$ and $\mathbf{P}_{\textbf{Max-SR}}^{(\mathbf{P})}$ presented below, we can obtain a locally optimal solution to the relaxed problem $\mathbf{P}_{\textbf{Max-SR}}^\prime$.
\begin{eqnarray}
\textbf{P}_{\textbf{Max-SR}}^{(\mathbf{P})}:  \quad 
\underset{{\mathbf{P}}}{\max} \quad (\ref{eq:obj_sr}) \quad 
\textbf{s.t.} \quad (\ref{eq:c_max_pwr}), \nonumber
\end{eqnarray}
\begin{eqnarray}
\mathbf{P}_{\textbf{Max-SR}}^{(\mathbf{F})}: \quad \underset{\mathbf{F}}{\max} \quad (\ref{eq:approx_obj_sr}) \quad
\textbf{s.t.} \quad (\ref{eq:c_max_res_per_user}), (\ref{eq:c_max_user_per_res}), (\ref{eq:c_max_pwr}), (\ref{eq:c_rlx_ch}), \nonumber 
\end{eqnarray}
where $\mathbf{F}^{\prime} = \mathbf{F}^{(t-1)}$, i.e the previous value of $\mathbf{F}$. 
\end{theorem}
\begin{proof}
From Lemma \ref{lemma:maxsr_cycle}, the feasible set of both $\mathbf{P}_{\textbf{Max-SR}}^{(\mathbf{F})}$ and $\mathbf{P}_{\textbf{Max-SR}}^{(\mathbf{P})}$ are convex. Moreover, from Lemma \ref{lemma:maxsr_approx}, we have that (\ref{eq:approx_obj_sr}) is a concave lower bound approximation to (\ref{eq:obj_sr_pen}) satisfying the conditions in (\ref{eq:ass1})-(\ref{eq:ass4}). Therefore, from the result shown in Theorem 2 in\cite{razaviyayn.2013}, the solution obtained by iteratively updating $\mathbf{F}$ and $\mathbf{P}$ cyclically is a local optimum of $\mathbf{P}_{\textbf{Max-SR}}^\prime$.
\end{proof}
As problems $\mathbf{P}_{\textbf{Max-SR}}^{(\mathbf{F})}$ and $\mathbf{P}_{\textbf{Max-SR}}^{(\mathbf{P})}$ are concave maximizations over a convex set and are readily solvable. Algorithm \ref{algo:max_sr} shows the pseudocode of the Max-SR algorithm, using $\mathbf{P}_{\textbf{Max-SR}}^{(\mathbf{F})}$ and $\mathbf{P}_{\textbf{Max-SR}}^{(\mathbf{P})}$ as subroutines.

\begin{algorithm}[!t]
\label{algo:max_sr}
\SetAlgoLined
\SetKw{Variables}{Variable Definition}
\SetKw{Init}{Initialization}
\SetKw{Out}{Output}
\SetKwFor{ReturnForAll}{return for all}{do}{}
\Variables{
\begin{itemize}[leftmargin=0.5cm]
\item[1.] $\mathbf{F}^{(t)}$ is the subcarrier allocation matrix at the $t$-th iteration.
\item[2.] $\mathbf{P}^{(t)}$ is the power allocation matrix at the $t$-th iteration.
\end{itemize}
}
\Init{
\begin{itemize}[leftmargin=0.5cm]
\item[1.] Set the initial values for the power allocation matrix $\mathbf{P}^{(0)}$ randomly, within the set defined by constraint (\ref{eq:c_max_pwr}), and, the subcarrier allocation matrix $\mathbf{F}^{(0)}$ within the set defined by constraints (\ref{eq:c_max_res_per_user}), (\ref{eq:c_max_user_per_res}) and (\ref{eq:c_rlx_ch}).
\item[2.] Set the convergence tolerance for the subcarrier allocation $\epsilon_F$ and for the power allocation $\epsilon_P$.
\item[3.] $t \leftarrow 0$
\end{itemize} 
}
\Out{
\begin{itemize}[leftmargin=.5cm]
\item[1.] Optimized power allocation $\mathbf{P}^*$.
\item[2.] Optimized subcarrier allocation $\mathbf{F}^*$.
\end{itemize}
}
\While{$\norm{\mathbf{F}^{(t)} - \mathbf{F}^{(t-1)}} > \epsilon_F $ and $\norm{\mathbf{P}^{(t)} - \mathbf{P}^{(t-1)}} > \epsilon_P $}{
	$t \leftarrow t + 1$; \\
    $\mathbf{F^{(t)}} \leftarrow \arg \mathbf{P}^{(\mathbf{F})}_{\textbf{Max-SR}}\left(\mathbf{F}^{(t-1)}, \mathbf{P}^{(t-1)} \right)$; (see Theorem \ref{th:max_sr})\\
    $\mathbf{P^{(t)}} \leftarrow \arg \mathbf{P}^{(\mathbf{P})}_{\textbf{Max-SR}}\left(\mathbf{F}^{(t-1)}, \mathbf{P}^{(t-1)} \right)$; (see Theorem \ref{th:max_sr})
}
$\mathbf{P}^* \leftarrow \mathbf{P}^{(t)}$; \\
$\mathbf{F}^* \leftarrow \mathbf{F}^{(t)}$; 
\caption{Maximization of sum-rate}
\end{algorithm}

\subsection{Solving $\mathbf{P}_{\textbf{Max-Min}}$}

Before solving the problem, notice that its objective function in (\ref{eq:obj_mm}) can be rewritten as
\begin{eqnarray}
\label{eq:non_concave_obj}
&\underset{j \in \mathcal{J}}{\min}& \underset{k \in \mathcal{K}}{\sum} \ln \left( 1 + \frac{|h_{k,j}|^2 f_{k,j} p_{k,j}}{\sigma_n^2 + \overset{j-1}{\underset{i = 1}{\sum}}|h_{k,i}|^2 f_{k,i} p_{k,i}}  \right) \nonumber \\
=& \underset{j \in \mathcal{J}}{\min}& \left[ \underset{k \in \mathcal{K}}{\sum} \ln \left(\sigma_n^2 + \overset{j}{\underset{i = 1}{\sum}} |h_{k,i}|^2 f_{k,i} p_{k,i} \right) -  \right. \nonumber \\
&&  \left. \ln \left(\sigma_n^2 + \overset{j-1}{\underset{i = 1}{\sum}} |h_{k,i}|^2 f_{k,i} p_{k,i}  \right) \right].
\end{eqnarray}
Both expressions inside the minimum function, $\ln \left(\sigma_n^2 + \overset{j}{\underset{i = 1}{\sum}} |h_{k,i}|^2 f_{k,i} p_{k,i} \right)$ and $\ln \left(\sigma_n^2 + \overset{j-1}{\underset{i = 1}{\sum}} |h_{k,i}|^2 f_{k,i} p_{k,i}  \right)$ are concave functions, which implies that their difference is non-concave. The summation of non-concave functions is also non-concave, and, by function composition rules \cite{boyd.2004}, the minimum of a non-concave function is non-concave as well.

Similarly to $\mathbf{P}_{\textbf{Max-SR}}$, the first step in solving $\mathbf{P}_{\textbf{Max-Min}}$ is relaxing the integer constraint on the entries of $\mathbf{F}$ and adding the same penalty function to its objective, leading to problem $\mathbf{P}_{\textbf{Max-Min}}^\prime$, given below
\begin{eqnarray}
&&\mathbf{P}_{\textbf{Max-Min}}^\prime: \nonumber \\
\label{eq:obj_mm_pen}
&\underset{{\mathbf{P}, \mathbf{F}}}{\max}&  
\underset{j \in \mathcal{J}}{\min} \left[ \underset{k \in \mathcal{K}}{\sum} \ln \left(\sigma_n^2 + \overset{j}{\underset{i = 1}{\sum}} |h_{k,i}|^2 f_{k,i} p_{k,i} \right) -  \right. \nonumber \\
&&  \left. \ln \left(\sigma_n^2 + \overset{j-1}{\underset{i = 1}{\sum}} |h_{k,i}|^2 f_{k,i} p_{k,i}  \right) \right] + \gamma \left(\mathbf{F} \right) \\
&\textbf{s.t.}& (\ref{eq:c_max_res_per_user}), (\ref{eq:c_max_user_per_res}), (\ref{eq:c_max_pwr}), (\ref{eq:c_rlx_ch}). \nonumber
\end{eqnarray}
After the relaxation, we have the non-concave term from (\ref{eq:non_concave_obj}) and the convex penalty function in the objective. The challenges involved in solving $\mathbf{P}_{\textbf{Max-Min}}^\prime$ are
\begin{itemize}
\item Just as in $\mathbf{P}_{\textbf{Max-SR}}^\prime$, the presence of multi-linear terms of the form $f_{k,j} p_{k,j}$ in (\ref{eq:obj_mm_pen}) and (\ref{eq:c_max_pwr}).
\item Even if $\mathbf{F}$ and $\mathbf{P}$ are updated cyclically, the objective function in (\ref{eq:obj_sr_pen}) is non-concave.
\end{itemize}
To address these issues, we follow the procedure used to solve $\mathbf{P}_{\textbf{Max-SR}}^\prime$, with a few extra steps. Firstly, notice that the feasible sets of $\mathbf{P}_{\textbf{Max-SR}}^\prime$ and $\mathbf{P}_{\textbf{Max-Min}}^\prime$ are the same, and therefore Lemma \ref{lemma:maxsr_cycle} also holds. We introduce a concave lower bound to (\ref{eq:obj_mm_pen}) in the Lemma below.

\begin{lemma}
\label{lemma:maxmin_approx_F}
Let $\theta_j \left( \mathbf{F} \right) = \underset{k \in \mathcal{K}}{\sum} \ln \left(\sigma_n^2 + \overset{j-1}{\underset{i = 1}{\sum}} |h_{k,i}|^2 f_{k,i} p_{k,i} \right)$. The function
\begin{eqnarray}
\label{eq:maxmin_approx_F}
&& \underset{j \in \mathcal{J}}{\min} \left[ \underset{k \in \mathcal{K}}{\sum} \ln \left(\sigma_n^2 + \overset{j}{\underset{i = 1}{\sum}}|h_{k,i}|^2 f_{k,i} p_{k,i} \right) -  \nonumber \right. \\ 
&&\left. \vphantom{\overset{j}{\underset{i = 1}{\sum}}} \theta_j \left( \mathbf{F}^{\prime} \right) - \trace \left( \nabla \theta_j \left( \mathbf{F}^{\prime} \right)^T \left( \mathbf{F} - \mathbf{F}^{\prime}\right) \right) \right] + \nonumber \\
&& \gamma \left( \mathbf{F^\prime} \right) + \trace \left[ \nabla \gamma\left(\mathbf{F}^{\prime} \right)^T \left(\mathbf{F} - \mathbf{F}^{\prime} \right)\right],
\end{eqnarray}
where $\nabla \theta_j \left( \mathbf{F}^{\prime} \right) \in \mathbb{R}^{K \times J}$ is a matrix such that
\begin{eqnarray}
\label{eq:cvxfy_F_grad}
\left[ \nabla \theta_j \left( \mathbf{F}^{\prime} \right) \right]_{k,n} &=& \left. \frac{\partial \theta_j \left( \mathbf{F} \right)}{\partial f_{k,n}} \right \vert_{\mathbf{F} = \mathbf{F^\prime}}  \nonumber \\ 
&=&  
\begin{cases}
\frac{|h_{k,n}|^2 p_{k,n}}{\sigma_n^2 + \overset{j-1}{\underset{i = 1}{\sum}}|h_{k,i}|^2 f^{\prime}_{k,i} p_{k,i}} &\text{ , } \forall \text{ } n < j \\
0 &\text{ , otherwise}
\end{cases}, \nonumber
\end{eqnarray}
is a lower bound concave approximation of (\ref{eq:obj_mm_pen}) for fixed $\mathbf{P}$ in the neighborhood of $\mathbf{F^\prime}$ which satisfies conditions (\ref{eq:ass1})-(\ref{eq:ass4}).
\end{lemma}
\begin{proof}
Firstly, as $\theta_j \left( \mathbf{F} \right)$ is a concave function we have that
\begin{equation*}
\theta_j \left( \mathbf{F} \right) \leq \theta_j \left( \mathbf{F}^{\prime} \right) - \trace \left( \nabla \theta_j \left( \mathbf{F}^{\prime} \right)^T \left( \mathbf{F} - \mathbf{F}^{\prime}\right) \right)
\end{equation*}
Therefore, the argument of the minimum function in (\ref{eq:maxmin_approx_F}) is less than the argument of the minimum function in (\ref{eq:obj_mm_pen}), which implies that (\ref{eq:ass2}) holds. Furthermore, (\ref{eq:maxmin_approx_F}) is obtained by approximating the non concave terms in (\ref{eq:obj_mm_pen}) by their first order linear approximation in the neighborhood of $\mathbf{F^\prime}$ and (\ref{eq:maxmin_approx_F}) is continuous, and hence (\ref{eq:ass1}), (\ref{eq:ass3}), and (\ref{eq:ass4}) also hold. Finally, (\ref{eq:maxmin_approx_F}) consists of the minimum of the summation over a concave function plus an affine term. So, 
\begin{eqnarray}
&\underset{k \in \mathcal{K}}{\sum}& \ln \left(\sigma_n^2 + \overset{j}{\underset{i = 1}{\sum}}|h_{k,i}|^2 f_{k,i} p_{k,i} \right) - \nonumber \\
&&\theta_j \left( \mathbf{F}^{\prime} \right) - \trace \left( \nabla \theta_j \left( \mathbf{F}^{\prime} \right)^T \left( \mathbf{F} - \mathbf{F}^{\prime}\right) \right), \nonumber
\end{eqnarray}
is concave. By composition rules \cite{boyd.2004}, the minimum of concave functions is also concave, which proves that (\ref{eq:maxmin_approx_F}) is concave, and coimpletes the proof.
\end{proof}

\begin{lemma}
\label{lemma:maxmin_approx_P}
Let $\theta_j \left( \mathbf{P} \right) = \underset{k \in \mathcal{K}}{\sum} \ln \left(\sigma_n^2 + \overset{j-1}{\underset{i = 1}{\sum}} |h_{k,i}|^2 f_{k,i} p_{k,i} \right)$. The function
\begin{eqnarray}
\label{eq:maxmin_approx_P}
&& \underset{j \in \mathcal{J}}{\min} \left[ \underset{k \in \mathcal{K}}{\sum} \ln \left(\sigma_n^2 + \overset{j}{\underset{i = 1}{\sum}}|h_{k,i}|^2 f_{k,i} p_{k,i} \right) -  \nonumber \right. \\ 
&&\left. \vphantom{\overset{j}{\underset{i = 1}{\sum}}} \theta_j \left( \mathbf{P}^{\prime} \right) - \trace \left( \nabla \theta_j \left( \mathbf{P}^{\prime} \right)^T \left( \mathbf{P} - \mathbf{P}^{\prime}\right) \right) \right] + \nonumber \\
&& \gamma \left( \mathbf{P^\prime} \right) + \trace \left[ \nabla \gamma\left(\mathbf{P}^{\prime} \right)^T \left(\mathbf{P} - \mathbf{P}^{\prime} \right)\right],
\end{eqnarray}
where $\nabla \theta_j \left( \mathbf{P}^{\prime} \right) \in \mathbb{R}^{K \times J}$ is a matrix such that
\begin{eqnarray}
\label{eq:cvxfy_P_grad}
\left[ \nabla \theta_j \left( \mathbf{P}^{\prime} \right) \right]_{k,n} &=& \left. \frac{\partial \theta_j \left( \mathbf{P} \right)}{\partial p_{k,n}} \right \vert_{\mathbf{P} = \mathbf{P^\prime}}  \nonumber \\ 
&=&  
\begin{cases}
\frac{|h_{k,n}|^2 f_{k,n}}{\sigma_n^2 + \overset{j-1}{\underset{i = 1}{\sum}}|h_{k,i}|^2 f^{\prime}_{k,i} p_{k,i}} &\text{ , } \forall \text{ } n < j \\
0 &\text{ , otherwise}
\end{cases}, \nonumber
\end{eqnarray}
is a lower bound concave approximation of (\ref{eq:obj_mm_pen}) for fixed $\mathbf{F}$ in the neighborhood of $\mathbf{P^\prime}$ which satisfies conditions (\ref{eq:ass1})-(\ref{eq:ass4}).
\end{lemma}
\begin{proof}
See the proof of Lemma \ref{lemma:maxmin_approx_F}.
\end{proof}
With the results from Lemmas \ref{lemma:maxsr_cycle}, \ref{lemma:maxmin_approx_F} and \ref{lemma:maxmin_approx_P} we can establish the convergence to a local optimum of an algorithm to solve $\mathbf{P}_{\textbf{Max-Min}}^\prime$.

\begin{theorem}
\label{th:max_min}
By updating $\mathbf{F}$ and $\mathbf{P}$ cyclically with the solutions to $\mathbf{P}_{\textbf{Max-Min}}^{(\mathbf{F})}$ and $\mathbf{P}_{\textbf{Max-Min}}^{(\mathbf{P})}$ presented below, we can obtain a locally optimal solution to the relaxed problem $\mathbf{P}_{\textbf{Max-Min}}^\prime$.
\begin{eqnarray}
\textbf{P}_{\textbf{Max-Min}}^{(\mathbf{P})}:  \quad 
\underset{{\mathbf{P}}}{\max} \quad (\ref{eq:maxmin_approx_P}) \quad 
\textbf{s.t.} \quad (\ref{eq:c_max_pwr}), \nonumber
\end{eqnarray}
where $\mathbf{P}^{\prime}$ is the value of $\mathbf{P}$ after the previous update. 
\begin{eqnarray}
\mathbf{P}_{\textbf{Max-Min}}^{(\mathbf{F})}: \quad \underset{\mathbf{F}}{\max} \quad (\ref{eq:maxmin_approx_F}) \quad
\textbf{s.t.} \quad (\ref{eq:c_max_res_per_user}), (\ref{eq:c_max_user_per_res}), (\ref{eq:c_max_pwr}), (\ref{eq:c_rlx_ch}), \nonumber 
\end{eqnarray}
where $\mathbf{F}^{\prime}$ is the value of $\mathbf{F}$ after the previous update. 
\end{theorem}
\begin{proof}
From Lemma \ref{lemma:maxsr_cycle}, the domains of $\mathbf{P}_{\textbf{Max-SR}}^{(\mathbf{F})}$ and $\mathbf{P}_{\textbf{Max-SR}}^{(\mathbf{P})}$ are convex. Also, from Lemmas \ref{lemma:maxmin_approx_F} and \ref{lemma:maxmin_approx_P}, we have that (\ref{eq:maxmin_approx_F}) and (\ref{eq:maxmin_approx_P}) are concave lower bound approximations to (\ref{eq:obj_sr_pen}) satisfying the conditions in (\ref{eq:ass1})-(\ref{eq:ass4}). Therefore, from the result in \cite{razaviyayn.2013}, the solution obtained by iteratively updating $\mathbf{F}$ and $\mathbf{P}$ cyclically is a local optimum of $\mathbf{P}_{\textbf{Max-SR}}^\prime$.
\end{proof}
The complete algorithm to solve the Min-Max problem, $\textbf{P}_{\textbf{Max-Min}}$, is described in Algorithm \ref{algo:max_min}. The algorithm uses $\mathbf{P}_{\textbf{Max-SR}}^{(\mathbf{F})}$ and $\mathbf{P}_{\textbf{Max-SR}}^{(\mathbf{P})}$ as subroutines.

\begin{algorithm}[!t]
\label{algo:max_min}
\SetAlgoLined
\SetKw{Variables}{Variable Definition}
\SetKw{Init}{Initialization}
\SetKw{Out}{Output}
\SetKwFor{ReturnForAll}{return for all}{do}{}
\Variables{
\begin{itemize}[leftmargin=0.5cm]
\item[1.] $\mathbf{F}^{(t)}$ is the subcarrier allocation matrix at the $t$-th iteration.
\item[2.] $\mathbf{P}^{(t)}$ is the power allocation matrix at the $t$-th iteration.
\end{itemize}
}
\Init{
\begin{itemize}[leftmargin=0.5cm]
\item[1.] Set the initial values for the power allocation matrix $\mathbf{P}^{(0)}$ randomly, within the set defined by constraint (\ref{eq:c_max_pwr}), and, the subcarrier allocation matrix $\mathbf{F}^{(0)}$ within the set defined by constraints (\ref{eq:c_max_res_per_user}), (\ref{eq:c_max_user_per_res}) and (\ref{eq:c_rlx_ch}).
\item[2.] Set the convergence tolerance for the subcarrier allocation $\epsilon_F$ and for the power allocation $\epsilon_P$.
\item[3.] $t \leftarrow 0$
\end{itemize} 
}
\Out{
\begin{itemize}[leftmargin=.5cm]
\item[1.] Optimized power allocation $\mathbf{P}^*$.
\item[2.] Optimized subcarrier allocation $\mathbf{F}^*$.
\end{itemize}
}
\While{$\norm{\mathbf{F}^{(t)} - \mathbf{F}^{(t-1)}} > \epsilon_F $ and $\norm{\mathbf{P}^{(t)} - \mathbf{P}^{(t-1)}} > \epsilon_P $}{
	$t \leftarrow t + 1$; \\
    $\mathbf{F^{(t)}} \leftarrow \arg \mathbf{P}^{(\mathbf{F})}_{\textbf{Max-Min}}\left(\mathbf{F}^{(t-1)}, \mathbf{P}^{(t-1)} \right)$; (see Theorem \ref{th:max_min})\\
    $\mathbf{P^{(t)}} \leftarrow \arg \mathbf{P}^{(\mathbf{P})}_{\textbf{Max-Min}}\left(\mathbf{F}^{(t-1)}, \mathbf{P}^{(t-1)} \right)$; (see Theorem \ref{th:max_min})
}
$\mathbf{P}^* \leftarrow \mathbf{P}^{(t)}$; \\
$\mathbf{F}^* \leftarrow \mathbf{F}^{(t)}$; 
\caption{Fairness Maximization}
\end{algorithm}

\section{Numerical Results}
\label{sec:results}
In this section the performance of the algorithms proposed in Section \ref{sec:analysis} is presented. Additionally, we compare our results with the three algorithms proposed in \cite{dabiri.2018}: the fixed user order (FUO), opportunistic allocation (OA) and proportional fair (PF) algorithms. In these three algorithms, the users are sorted according to a different criteria, and the resources are allocated sequentially by picking the best available resource for the user in the sorted order. The FUO algorithm performs the allocation in a random order, while the OA algorithm sorts the users according to their overall channel qualities, prior to the channel allocation. The PF algorithm takes into account the $L$ past channel qualities when sorting the allocation order in order to improve fairness. In our evaluation we consider $L=10$ and $\alpha = 0.9$.

We consider a scenario where one BS is serving $6$ users over $4$ subcarriers, with $N=2$ and $d_f=3$, in a circular cell of radius $R=300$ m and the users are uniformly distributed inside the cell. It is worth mentioning that an increase in $N$ would result in higher diversity, as each user would transmit its signal on more subcarriers. However, as the value of $d_f$ is tied to $N$ (i.e. $d_f = {{K-1}\choose{N-1}}$), $d_f$ would also increase, resulting in an exponential increase in the decoding complexity. We consider a path loss exponent $\alpha = 4$. We consider a noise power density of $-174 \text{ dBm/Hz}$ and a bandwidth of $180 \text{ kHz}$. Also, we consider a normalized slow fading Rayleigh channel, such that the channel remains constant for the duration of each transmitted symbol. Furthermore, we simulate the algorithms' performance for a maximum transmit power per user varying between $3 \text{ dBm}$ and $10 \text{ dBm}$. We evaluate the performances according to two metrics: the sum-rate and the Jain's fairness index \cite{lan.2010}. The former, is a measure of the overall achievable throughput of the network and the latter is a measure of the fairness of the resource allocation between the users based on their individual achievable throughputs. Let $\mathbf{c} \in \mathbb{R}^{J}$ be a vector, such that, its $i$-th coordinate, $c_i$, corresponds to the throughput of the $i$-th user. The Jain's fairness index for a given rate vector, $\mathbf{c} = [c_1, \cdots, r_J]^T$, is 
\begin{equation}
J(\mathbf{c}) = \frac{\left( \underset{j=1}{\overset{J}{\sum}} c_j \right)^2}{J \underset{j=1}{\overset{J}{\sum}} c_j^2}.
\end{equation}
This index varies from $\frac{1}{J}$, meaning no fairness, to $1$, meaning perfect fairness. Furthermore, we consider a normalized Rayleigh fading channel, and the performances are averaged over several channel realizations.

To implement the proposed algorithms, we used the convex optimization modeling language CVXPY \cite{cvxpy,cvxpy_rewriting}, with the open-source ECOS solver \cite{domahidi.2013}, due to its support of exponential cones \cite{akle.2015}. The comparison of the sum-rate of the five algorithms is shown in Figure \ref{fig:sum-rate}. The sum-rate of the Max-SR algorithm outperforms the Max-Min, FUO, OA, and PF algorithms for the whole range of transmitted power evaluated. The three algorithms proposed in \cite{dabiri.2018} present similar performances with the proportional fairness being slightly worse than the other two. Finally, the Max-Min algorithm is greatly outperformed by the other ones. This result is expected since the Max-Min gives up on maximizing the sum-rate in favor of improving the fairness.
\begin{figure}
\centering {
  \ifCLASSOPTIONtwocolumn 
		\hspace*{-.7cm}\includegraphics[width=1.15\columnwidth]{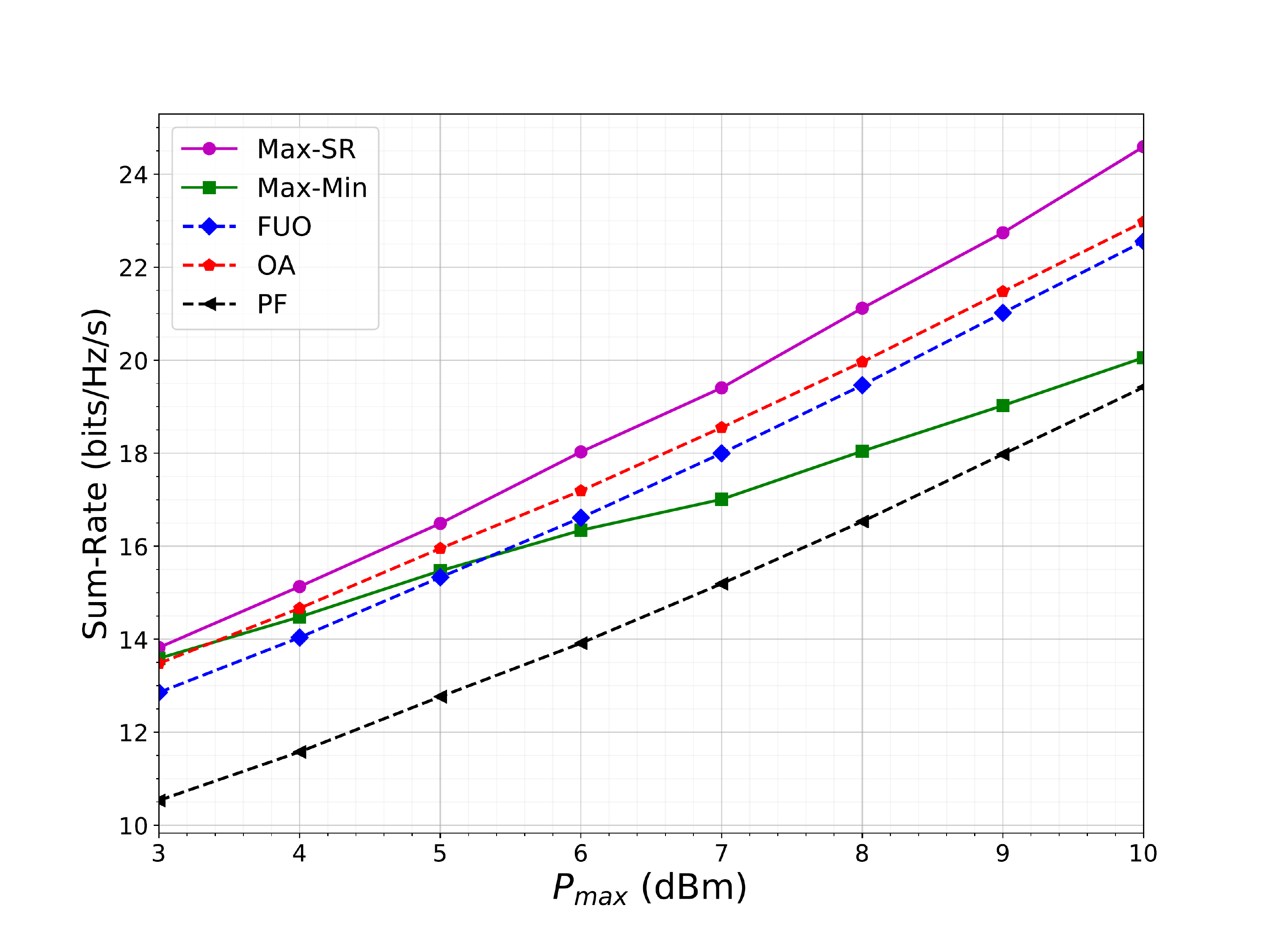}\vspace*{-0.6cm}
  \else
    	\includegraphics[width=0.7\columnwidth]{figs/sumrate_comp-eps-converted-to.pdf}
  \fi

\caption{Sum-rate comparison for $J=6$, $K=4$, $d_f=3$ and $N=2$. \label{fig:sum-rate}}
}
\end{figure}

Figure \ref{fig:jain} shows the Jain's fairness index achieved by each algorithm. The Max-Min algorithm greatly outperforms the alternatives for the whole range of transmitted powers. It is worth noting that the Jain's fairness index is bottlenecked by the rate of the user with the worst channel. Therefore, an increase in the maximum transmit power results in a higher throughput for the worst user, consequently, increasing the overall fairness. Furthermore, with increasing maximum transmit power, the fairness of the Max-Min algorithm approaches one. The other algorithms achieve similar fairness performance, with the PF algorithm achieving slightly better fairness than the others.
\begin{figure}
\centering {
  \ifCLASSOPTIONtwocolumn 
		\hspace*{-.3cm}\includegraphics[width=1.15\columnwidth]{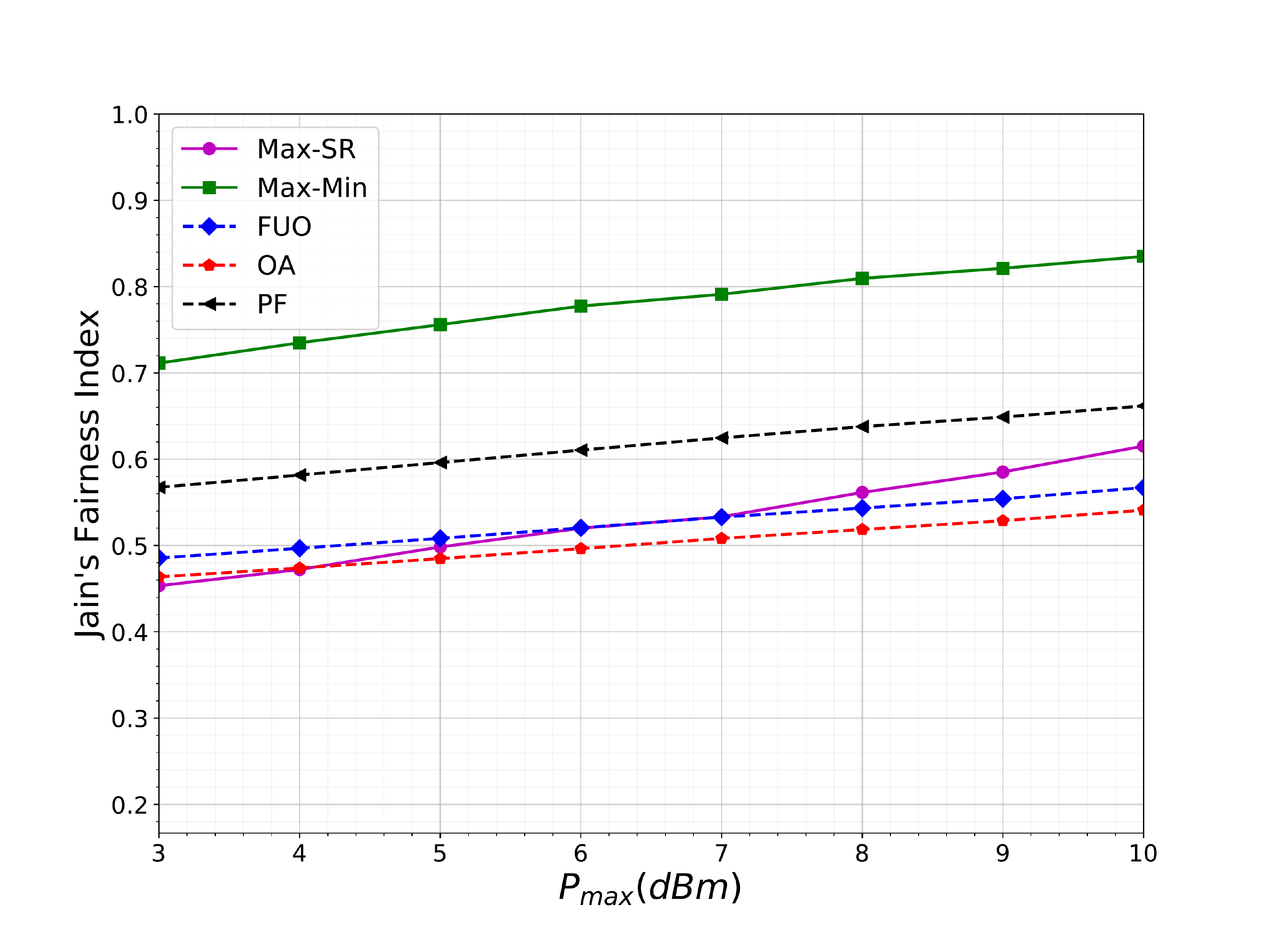}\vspace*{-0.6cm}
  \else
    	\includegraphics[width=.7\columnwidth]{figs/jain_comp-eps-converted-to.pdf}
  \fi

\caption{Jain's fairness index comparison for $J=6$, $K=4$, $d_f=3$ and $N=2$. \label{fig:jain}}
}
\end{figure}

In order to evaluate the link level performance of the Max-SR and Max-Min allocation, another simulation, evaluating the BER of both allocation algorithms, is presented. In the simulation, each user transmits a multidimensional symbol with $M = 4$ using a quadrature amplitude modulation (QAM) mother constellation \cite{nikopour.2013}. We assume that the users' channel gains are ordered, such that, $\norm{\mathbf{H}_1}_F \leq \norm{\mathbf{H}_2}_F \leq \cdots \leq \norm{\mathbf{H}_J}_F$, i.e., the first user has the worst channel gain, while the last user has the best one. Figure \ref{fig:BER}, shows a comparison between the BER of the $6$ users using Max-SR and Max-Min allocation. As expected, with Max-SR allocation, the error probabilities are small, but there is a large gap between the best and the worst user. On the other hand, Max-Min allocation results in overall larger error probabilities when compared to the Max-SR, but the performance gap between users with different channel quality is smaller.
\begin{figure}
\centering {
  \ifCLASSOPTIONtwocolumn 
		\hspace*{-.3cm}\includegraphics[width=1.15\columnwidth]{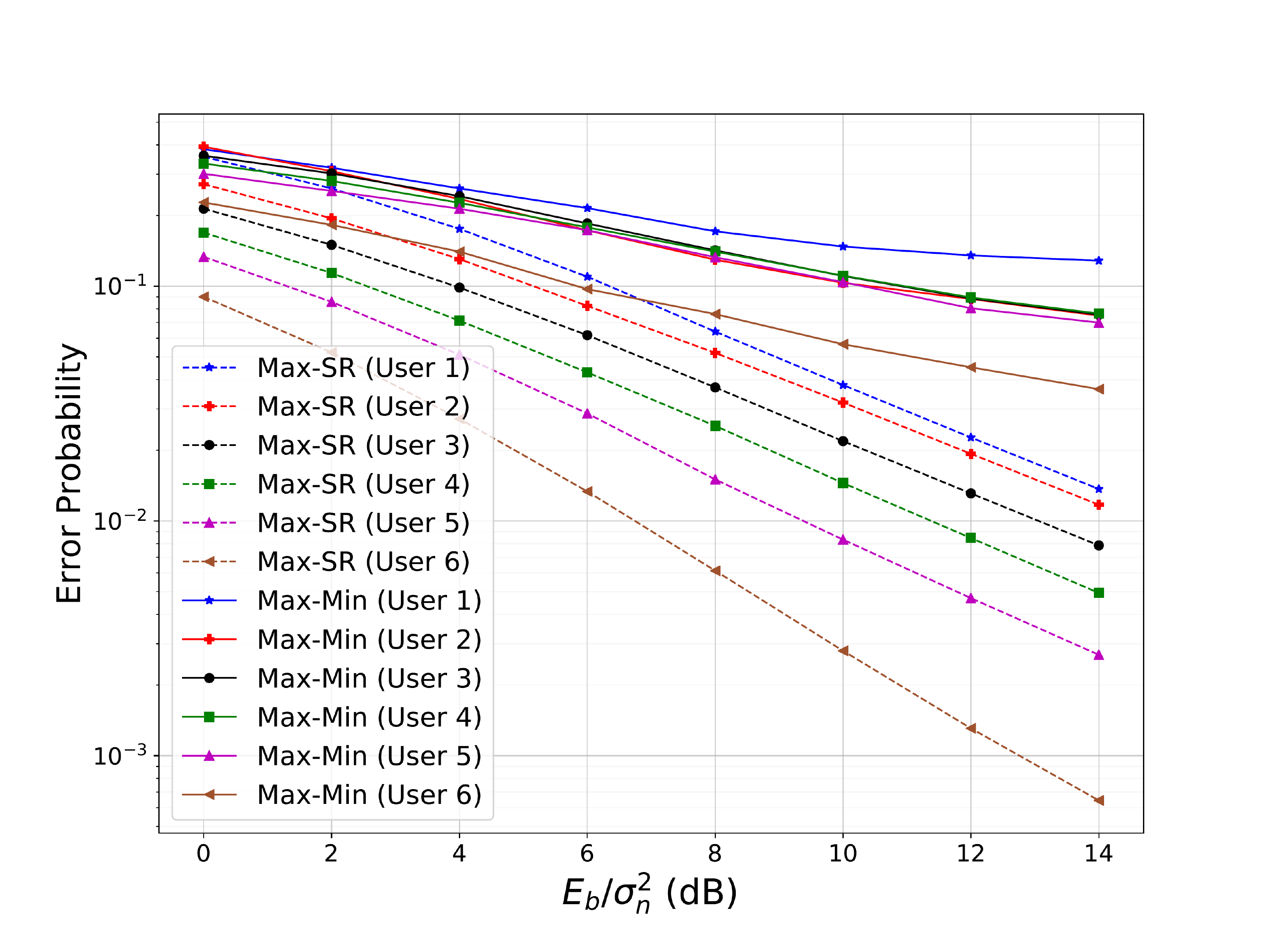}\vspace*{-0.6cm}
  \else
    	\includegraphics[width=0.7\columnwidth]{figs/BER_comparison-eps-converted-to.pdf}
  \fi

\caption{BER comparison for $J=6$, $K=4$, $d_f=3$ and $N=2$. \label{fig:BER}}
}
\end{figure}
\subsection{Performance with Outdated CSI}
\label{sec:outdated_csi}

\begin{figure}[!h]
    \centering {
    \ifCLASSOPTIONtwocolumn 
        \includegraphics[width=0.7\columnwidth]{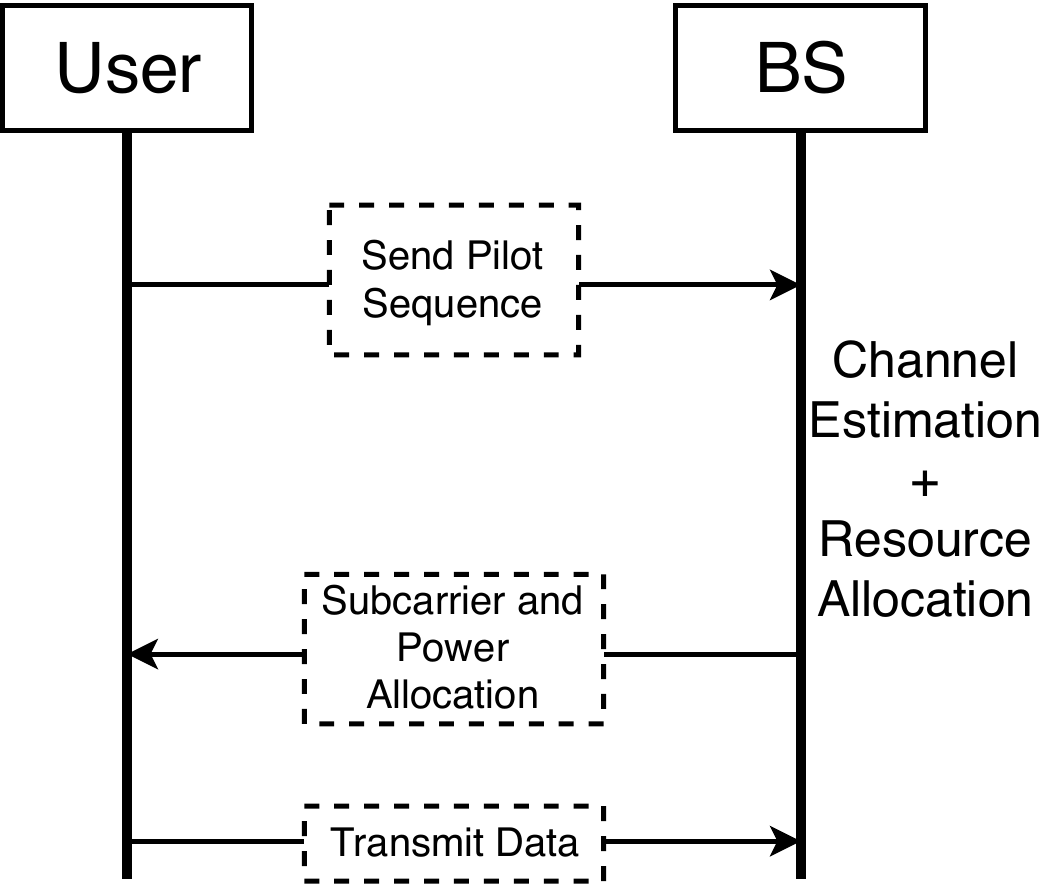}
    \else
        \includegraphics[width=0.4\columnwidth]{figs/architecture-eps-converted-to.pdf}
    \fi
    \caption{Resource allocation procedure}
    \label{fig:rap}
    }
\end{figure}

In this section, we investigate the effect on the performance of the algorithms, under an outdated CSI regime. In the results shown so far, we considered that for every new channel realization the users would send a pilot sequence to the BS, who would run the optimization routine and send the allocations back to the users. This approach requires a large overhead as it requires pilots and allocations to be sent constantly between the users and the BS. Hence, we consider a system where the pilots are sent periodically with period $T$, and the users reuse the same allocation during the period, as illustrated in Figure \ref{fig:rap}.
    
In this experiment, we model the temporal relationship between two successive channel realizations as an i.i.d first-order Gauss-Markov process \cite{Patzold2012MobileChannels}, for each $k \in \mathcal{K}$ and $j \in \mathcal{J}$, given by
\begin{equation}
    h^{(n+1)}_{k, j} = \rho h^{(n+1)}_{k, j} + w_{k, j} , 
    \end{equation}
    where $h^{(n)}_{k,j} \sim \mathcal{CN}(0,1)$ and $w^{(t)}_{k,j} \sim \mathcal{CN}(0,1-\rho^2)$ is the innovation component. Moreover The correlation between successive fading components is given by
\begin{equation}
\label{eq:fading_corr}
    \rho = J_0 \left(2 \pi f_{\max} T_s \right),
\end{equation}
where $f_{\max}$ is the maximum Doppler frequency, $T_s$ is the time between channel updates, and $J_0$ is the Bessel function. Figure \ref{fig:outdated_csi} shows the performance deterioration for $T \in [1, 50]$, $P_{\max} = 10$ dBm, $T_S = 0.01$ s, and four values of $f_{\max}$ resulting in $\rho^2 \in \{ 0.95, 0.62, 0.22, 0.01 \}$.
\begin{figure}
\centering {
  \ifCLASSOPTIONtwocolumn 
		\hspace*{-.3cm}\includegraphics[width=1\columnwidth]{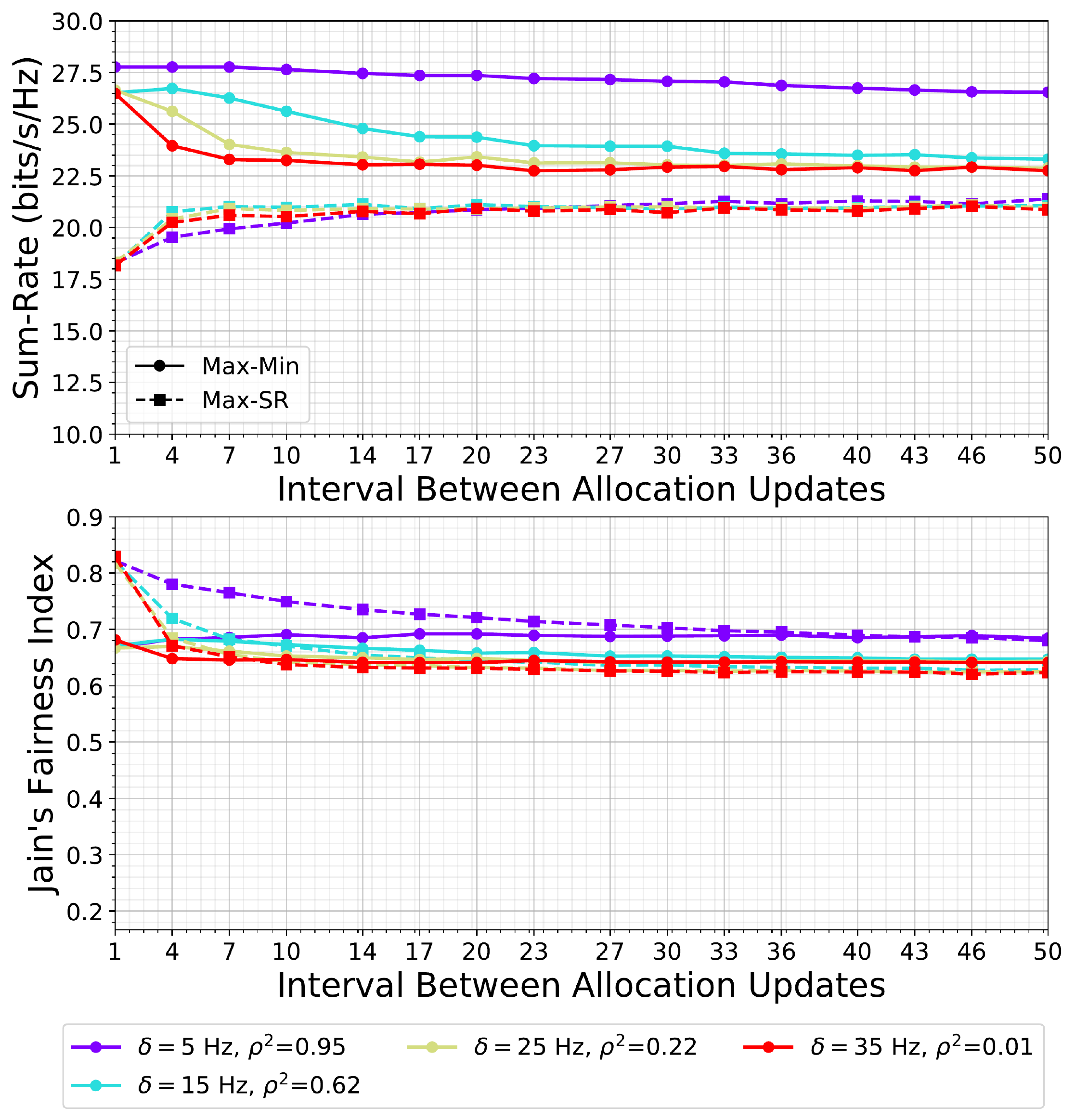}\vspace*{-0.1cm}
  \else
    	\includegraphics[width=0.7\columnwidth]{figs/outdated_csi-eps-converted-to.pdf}
  \fi
    \caption{\label{fig:outdated_csi}Performance of the Max-SR (solid line) and Max-Min (dashed line) allocations for $T \in [1, 50]$ and different values of $f_{\max}$.}} 
    \end{figure}
Tables \ref{tab:outdated_csi_maxsr} and \ref{tab:outdated_csi_maxmin} summarizes the effects of outdated CSI in the performance of Max-SR and Max-Min algorithms respectively. From the table results we conclude that the Max-SR algorithm is more robust to outdated CSI as for the worst case ($T=50$ and $\rho^2 = 0.01$), it still achieves $85\%$ of the sum-rate with $T=1$, while the Max-Min only achieves $75\%$ of the original fairness. Moreover, as in the proposed model the same allocation is reused for subsequent transmissions, and the less correlated the current CSI is with the one used to obtain the allocation (i.e. longer periods between updates), the more it resembles a random allocation, resulting in a performance decrease. For instance, the Max-Min algorithm maximizes the fairness of the system, hence, its fairness index decreases with longer periods. On the other hand, we observe in Figure \ref{fig:outdated_csi} that its sum-rate increases with longer periods, thus, we  conclude that the Max-Min algorithm increases fairness at the expense of the sum-rate. The same does not happen in the Max-SR algorithm as the fairness remains roughly the same under the effect of outdated CSI.

\begin{table}[!t]
\begin{center}
\caption{Performance Deterioration with Outdated CSI (Max-SR) \label{tab:outdated_csi_maxsr}}
\begin{tabular}{| c | c | c |}
\hline %
& Fairness & Sum-Rate \\
\hline
$\rho^2 = 0.95$ & $102.38$\% & $98.81$\%\\
\hline
$\rho^2 = 0.62$ & $96.45$\% & $87.88$\%\\
\hline
$\rho^2 = 0.22$ & $96.47$\% & $85.85$\%\\
\hline
$\rho^2 = 0.01$ & $94.11$\% & $85.89$\%\\
\hline
\end{tabular}
\end{center}
\end{table}
\begin{table}[!t]
\begin{center}
\caption{Performance Deterioration with Outdated CSI (Max-Min) \label{tab:outdated_csi_maxmin}}
\begin{tabular}{| c | c | c |}
\hline 
& Fairness & Sum-Rate \\
\hline
$\rho^2 = 0.95$ & $83.58$\% & $119.73$\%\\
\hline
$\rho^2 = 0.62$ & $77.22$\% & $114.51$\%\\
\hline
$\rho^2 = 0.22$ & $76.66$\% & $112.38$\%\\
\hline
$\rho^2 = 0.01$ & $74.99$\% & $115.34$\%\\
\hline
\end{tabular}
\end{center}
\end{table}

\subsection{Convergence Analysis and Algorithm Complexity}
\label{sec:conv_analysis}
In this section, the convergence rate of the proposed algorithms and their complexity is investigated.
The stopping criteria for the procedure is based on the difference between successive updates of the optimization variables $\mathbf{F}$ and $\mathbf{P}$. When this difference falls below a threshold, $\epsilon_F$ and $\epsilon_P$, respectively, the algorithm has converged. In order to illustrate the convergence, Algorithms \ref{algo:max_sr} and \ref{algo:max_min} are simulated with $P_{\text{max}}^{(j)} = 10$ dBm $\forall \text{ } j \in \mathcal{J}$, with the same channel gains and users location, but with three randomly chosen initial conditions $(\mathbf{F}^{(0)}, \mathbf{P}^{(0)})$. 
Figure \ref{fig:conv_maxsr} shows the convergence of Algorithm \ref{algo:max_sr}. Each iteration in Figure \ref{fig:conv_maxsr} consists of a full cycle of updates. 
\begin{figure}
\centering {
  \ifCLASSOPTIONtwocolumn 
		\hspace*{-.3cm}\includegraphics[width=1.1\columnwidth]{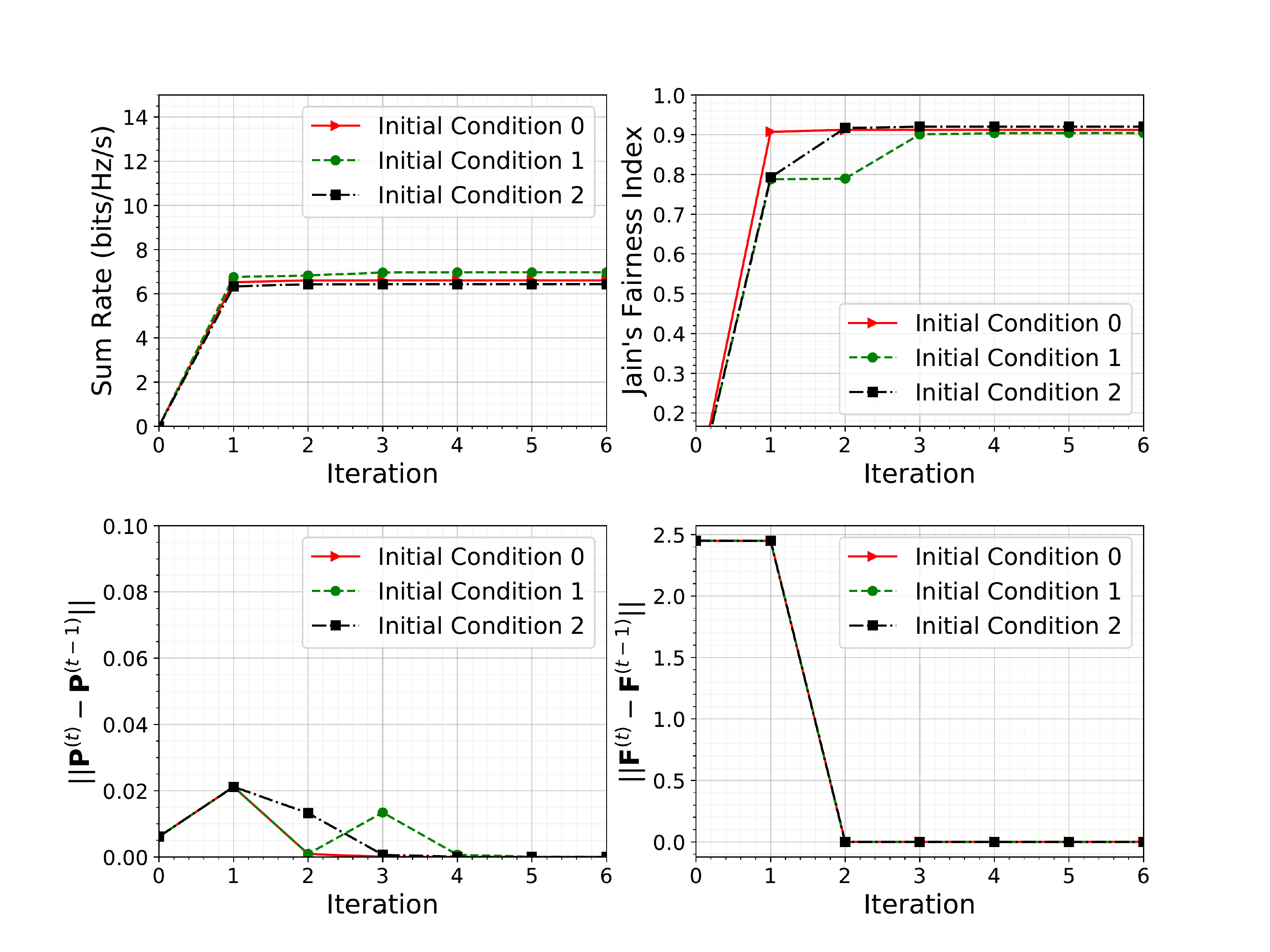}
  \else
    	\includegraphics[width=0.7\columnwidth]{figs/maxsr_convergence-eps-converted-to.pdf}
  \fi

\caption{Convergence of the $\textbf{P}_{\textbf{Max-SR}}$, with maximum transmit power of $10$ dBm, algorithm for three different initial conditions. \label{fig:conv_maxsr}}
}
\end{figure}
In Figure \ref{fig:conv_maxmin}, the convergence of Algorithm \ref{algo:max_min} is shown. The algorithm converges in five steps or less for all three initial conditions. Furthermore, 
Algorithm \ref{algo:max_sr} consistently converges to a higher sum-rate than Algorithm \ref{algo:max_min}, while Algorithm \ref{algo:max_min} converges to a higher fairness index as expected.
\begin{figure}
\centering {
  \ifCLASSOPTIONtwocolumn 
	\hspace*{-.3cm}\includegraphics[width=1.1\columnwidth]{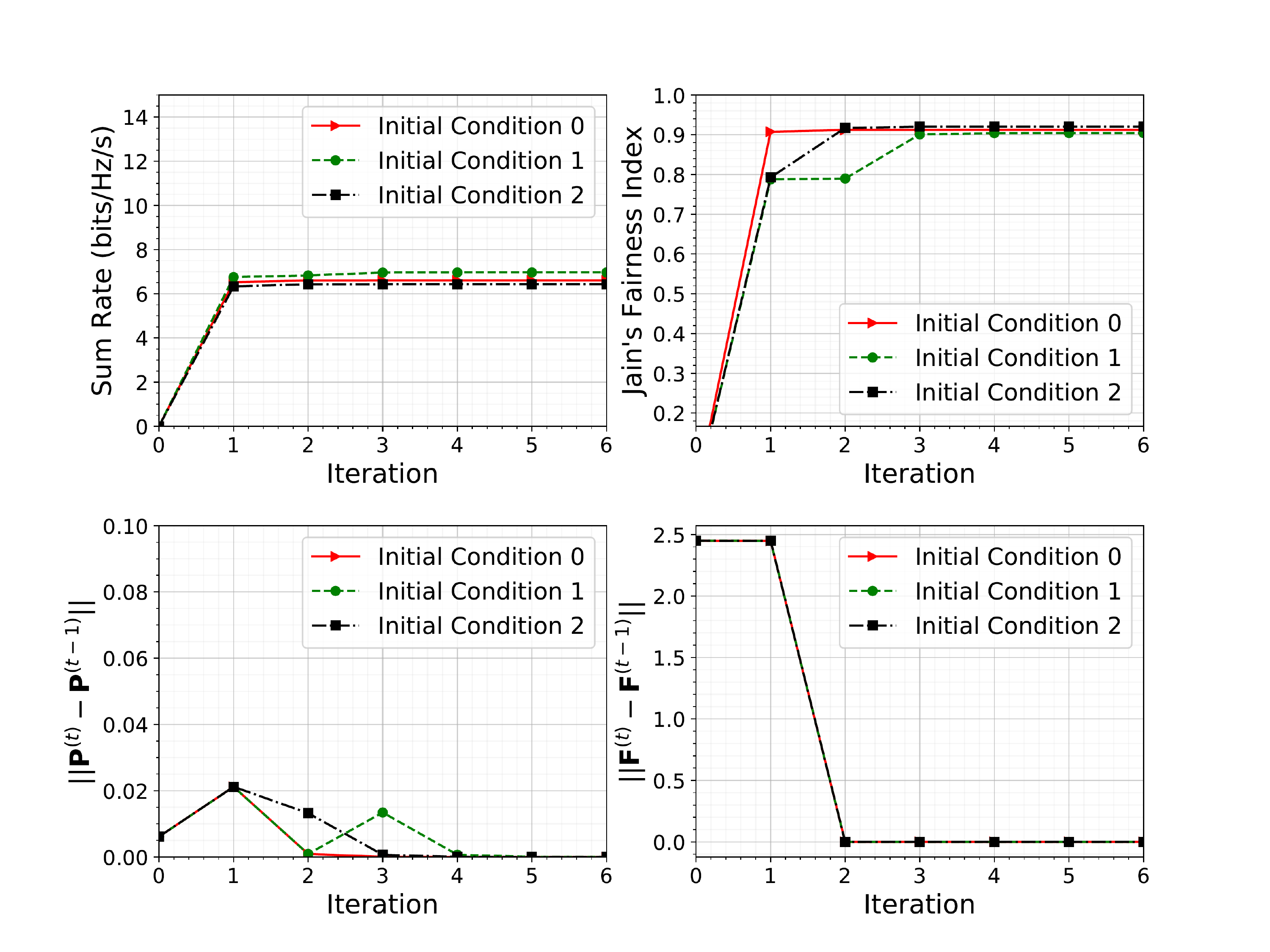}
  \else
    \includegraphics[width=.7\columnwidth]{figs/maxmin_convergence-eps-converted-to.pdf}
  \fi

\caption{Convergence of the $\textbf{P}_{\textbf{Max-Min}}$ algorithm, with maximum transmit power of $10$ dBm, for three different initial conditions. \label{fig:conv_maxmin}}
}
\end{figure}

\subsubsection{Algorithm Complexity}
Each update of the Max-SR algorithm involves solving a convex optimization problem, namely $\mathbf{P}_{\textbf{Max-SR}}^{(\mathbf{F})}$ for the subcarrier allocation update and  $\mathbf{P}_{\textbf{Max-SR}}^{(\mathbf{P})}$ for the power allocation update. We consider that both update steps are solved using a primal-dual interior point algorithm \cite{nocedal.2006} with a logarithmic barrier function.
\begin{definition}[$\epsilon$-optimal solution]
Let
\begin{eqnarray}
&\underset{\mathbf{x}}{\max}& f(\mathbf{x}) \quad \textbf{s.t.} \quad \mathbf{x} \in \mathcal{X}, \nonumber
\end{eqnarray}
where $f(\mathbf{x})$ is a concave function and $\mathcal{X}$ is a convex set be an optimization problem such that $\mathbf{x}^*$ is the unique maximizer to the problem. The vector $\mathbf{x}^\prime \in \mathcal{X}$ is an $\epsilon$-optimal solution to the problem if
\begin{equation}
    f(\mathbf{x}^*) - f(\mathbf{x}^\prime) \leq \epsilon.
\end{equation}
\end{definition}
As derived in \cite{nesterov.1994}, the number of steps required to obtain an $\epsilon$-optimal solution to a convex optimization problem using the interior point algorithm with logarithmic barrier function is $\mathcal{O}\left( \sqrt{n} \log_2 \left(\frac{1}{\epsilon}\right)\right)$, where $n$ is the number of inequality constraints. The problems solved in the update steps $\mathbf{P}_{\textbf{Max-SR}}^{(\mathbf{F})}$ and $\mathbf{P}_{\textbf{Max-SR}}^{(\mathbf{P})}$ have $2J + K + KJ$ and $J$ inequality constraints, respectively. Therefore, the complexity of each subcarrier allocation update is $\mathcal{O}\left( \sqrt{2J + K + KJ} \log_2 \left(\frac{1}{\epsilon}\right)\right)$ and the complexity of each power allocation update is $\mathcal{O}\left( \sqrt{J} \log_2 \left(\frac{1}{\epsilon}\right)\right)$. In order to establish the total number of iterations required for the convergence of the algorithm, we use the result from Theorem 3.1 in \cite{hong.2017}. Let
\begin{equation*}
    \underset{\mathbf{x}}{\max} \quad f(\mathbf{x}) \quad \textbf{s.t.} \quad \mathbf{x} \in \mathcal{X}, \nonumber
\end{equation*}
where $f(\mathbf{x})$ might be non-concave and non-smooth, be a generic non-convex optimization problem, and $\mathbf{x}^* \in \mathcal{X}^*$, where $\mathcal{X}^*$ is the set of the problem's stationary points. Then, the optimality gap after the $t$-th cyclic update of the BLSM algorithm is given by
\begin{equation}
    \Delta_t = f\left(\mathbf{x}^*\right) - f\left(\mathbf{x}^{(t)}\right) \leq \frac{c}{t},
\end{equation}
where $c$ is a constant. Therefore, to obtain an $\epsilon$-optimal solution, $\mathcal{O}\left(\frac{1}{\epsilon}\right)$ update steps are necessary. The same results apply to the Max-Min algorithm as the same number of constraints are involved to solve the power allocation and subcarrier allocation update problems. Both algorithms complexity are summarized on Table \ref{tab:complexity}. 

\begin{table}[!t]
\begin{center}
\caption{Complexity of the Max-SR and Max-Min algorithms \label{tab:complexity}}
\begin{tabular}{| c | c |}
\hline
Procedure & Number of Steps \\
\hline
Subcarrier Allocation Update & $\mathcal{O}\left( \sqrt{2J + K + KJ} \log_2 \left(\frac{1}{\epsilon}\right)\right)$\\
Power Allocation Update & $\mathcal{O}\left( \sqrt{J} \log_2 \left(\frac{1}{\epsilon}\right)\right)$  \\
Total Number of Updates & $\mathcal{O} \left( \frac{1}{\epsilon} \right)$ \\
\hline
\end{tabular}
\end{center}
\end{table}
\section{Conclusions}
\label{sec:conclusions}
In this paper, two joint channel and power allocation algorithms are proposed: the Max-SR and the Max-Min algorithms. The former aims for sum-rate  maximization, while the latter aims for maximizing fairness. The BSUM framework is employed to obtain algorithms converging to locally optimal points of the relaxed problems. We compare the performance of the Max-SR and the Max-Min algorithms with the ones proposed in \cite{dabiri.2018}. The results show that the Max-SR algorithm has better performance on the sum-rate sense, while the Max-Min has better performance on the fairness sense. Moreover, a numerical analysis of the convergence of the algorithms is presented. Finally, we derive the worst-case time complexity of both algorithms.
The results show that the Max-SR consistently achieve a better sum-rate, while the Max-Min achieves better fairness. 
Furthermore, there is a tradeoff between the fairness and the sum-rate. For future works, we intend to investigate new algorithms that can achieve a compromise between sum-rate and fairness. 

\appendices

\section{Proof of Theorem \ref{th:np_hard}} \label{app:np_proof}
To prove that both problems are NP-hard, we show that the subcarrier assignment subproblem can be reduced in polynomial time to the hypergraph assignment problem (HAP) which is shown to be NP-hard in \cite{borndofer.2015}.

We start by briefly introducing the HAP. The HAP takes as input a bipartite graph $\mathcal{G} = \mathcal{V} \cup \mathcal{U}$ such that $\mathcal{V} \cap \mathcal{U} = \emptyset$, a set of hyperedges $\mathcal{E}$, and a cost function $c:\mathcal{E}\rightarrow \mathbb{R}$. In addition to that, we have that $\left\vert e \cap \mathcal{V} \right\vert \geq 1 \text{ and }  \left\vert e \cap \mathcal{U} \right\vert \geq 1 \text{ } \forall \text{ } e \in \mathcal{E}$. 
A hyperassignment is a set $\mathcal{H} \subseteq \mathcal{E}$ such that every element $v \in \mathcal{V} \cup \mathcal{U}$ appears exactly once in $\mathcal{H}$. The output of the HAP is an optimal hyperassignment $\mathcal{H}^* = \min \{c(\mathcal{H}) \vert \mathcal{H} \text{ is an hyperassignment of } \mathcal{G} \}$.

Now consider the problem $\mathbf{P}_{\textbf{Max-SR}}$ in the HAP context, with a fixed power allocation matrix $\mathbf{P}$. Let $\mathcal{U} = \{SC_1, \cdots, SC_K\}$ and $\mathcal{V} = \{U_{1,1} \cdots U_{1,N}, U_{2,1}, \cdots U_{2, N}, \cdots U_{J,N}\}$ be the vertex set of available subcarriers and allocated subcarriers, respectively. Notice that $SC_k$ denotes the $k$-th available subcarrier, while $U_{j,i}$ denotes the $i$-th subcarrier allocated to user $j$, with $1 \leq i \leq N$. In this context, a hyperassignment determines which subcarriers are allocated to each user. For instance, a hyperedge $e = \{SC_1, U_{1,1}, U_{2,2}\}$ indicates that users $1$ and $2$ are allocated to the first subcarrier. As at most $N$ subcarriers can be allocated to each user, every hyperassignment satisfies \eqref{eq:c_max_res_per_user}. Notice that the subcarrier allocation matrix $\mathbf{F}$ corresponds to the hyperedge incidence matrix of the hypergraph, for instance, the hyperedge $e = \{SC_1, U_{1,1}, U_{2,2}\}$ would result in $\mathbf{f}_1 = \begin{bmatrix}
1 & 1 & 0 & \cdots & 0
\end{bmatrix}$, where $\mathbf{f}_{1}$ is the first row of $\mathbf{F}$. Furthermore, consider the hyperedge set $\mathcal{E} = \{e \in 2^{\mathcal{U} \cup \mathcal{V}}\vert \quad 1 < |e| \leq d_f+1\}$ and let $F_{\mathcal{H}} \in \mathbb{B}^{K \times J}$ be the incidence matrix of a hypermatch $\mathcal{H} \subseteq \mathcal{E}$. As $\vert e \vert \leq d_f+1 \text{ } \forall \text{ } e \in \mathcal{E}$ constraint \eqref{eq:c_max_user_per_res} is always satisfied. Furthermore, the cost function is given as
\begin{equation*}
    c(\mathcal{H}) = \underset{k \in \mathcal{K}}{\sum} \ln \left( 1 + \frac{\underset{j \in \mathcal{J}}{\sum} |h_{k, j}|^2 f_{k, j} p_{k, j}}{\sigma_n^2} \right).
\end{equation*}
So, solving $\mathbf{P}_{\textbf{Max-SR}}$ with $\mathbf{P}$ fixed, is equivalent to solving an HAP, therefore, $\mathbf{P}_{\textbf{Max-SR}}$ is NP-hard.

Finally, we also conclude $\mathbf{P}_{\textbf{Max-Min}}$ is NP-hard using the same HAP formulation, but using the cost function
\begin{equation*}
    c(\mathcal{H}) = \underset{j \in \mathcal{J}}{\min} \underset{k \in \mathcal{K}}{\sum} \ln \left( 1 + \frac{|h_{k,j}|^2 f_{k,j} p_{k,j}}{\sigma_n^2 + \overset{j-1}{\underset{i = 1}{\sum}}|h_{k,i}|^2 f_{k,i} p_{k,i}}  \right).
\end{equation*}

\ifCLASSOPTIONcaptionsoff
  \newpage
\fi
%
{\footnotesize{
\bibliographystyle{IEEEtran}
\bibliography{IEEEabrv,references}{}}

\begin{thebibliography}{10}
\providecommand{\url}[1]{#1}
\csname url@samestyle\endcsname
\providecommand{\newblock}{\relax}
\providecommand{\bibinfo}[2]{#2}
\providecommand{\BIBentrySTDinterwordspacing}{\spaceskip=0pt\relax}
\providecommand{\BIBentryALTinterwordstretchfactor}{4}
\providecommand{\BIBentryALTinterwordspacing}{\spaceskip=\fontdimen2\font plus
\BIBentryALTinterwordstretchfactor\fontdimen3\font minus
  \fontdimen4\font\relax}
\providecommand{\BIBforeignlanguage}[2]{{%
\expandafter\ifx\csname l@#1\endcsname\relax
\typeout{** WARNING: IEEEtran.bst: No hyphenation pattern has been}%
\typeout{** loaded for the language `#1'. Using the pattern for}%
\typeout{** the default language instead.}%
\else
\language=\csname l@#1\endcsname
\fi
#2}}
\providecommand{\BIBdecl}{\relax}
\BIBdecl

\bibitem{razaviyayn.2013}
M.~Razaviyayn, M.~Hong, and Z.-Q. Luo, ``A unified convergence analysis of
  block successive minimization methods for nonsmooth optimization,''
  \emph{SIAM Journal on Optimization}, vol.~23, no.~2, pp. 1126--1153, Jun
  2013.

\bibitem{ericsson}
``{Ericsson Mobility Report},''
  \url{https://www.ericsson.com/assets/local/mobility-report/documents/2015/ericsson-mobility-report-june-2015.pdf},
  Tech. Rep., 2015.

\bibitem{cover.1972}
T.~Cover, ``Broadcast channels,'' \emph{{IEEE} Trans. Inf. Theory}, vol.~18,
  no.~1, pp. 2--14, Jan 1972.

\bibitem{ahlswede.1973}
R.~Ahlswede, ``{Multi-way communication channels},'' in \emph{2nd Int. Symp. on
  Inform. Theory}.\hskip 1em plus 0.5em minus 0.4em\relax Akademiai Kiado, Sept
  1973, pp. 23--51.

\bibitem{moltafet.2018a}
M.~Moltafet, N.~Mokari, M.~R. Javan, H.~Saeedi, and H.~Pishro-Nik, ``A new
  multiple access technique for {5G}: Power domain sparse code multiple access
  ({PSMA}),'' \emph{{IEEE} Access}, vol.~6, pp. 747--759, Nov. 2018.

\bibitem{wang.2015}
B.~Wang, K.~Wang, Z.~Lu, T.~Xie, and J.~Quan, ``Comparison study of
  non-orthogonal multiple access schemes for {5G},'' in \emph{IEEE Int. Symp.
  on Broadband Multimedia Systems and Broadcasting}, June 2015, pp. 1--5.

\bibitem{moltafet.2018}
M.~Moltafet, N.~M. Yamchi, M.~R. Javan, and P.~Azmi, ``Comparison study between
  {PD-NOMA} and {SCMA},'' \emph{{IEEE} Trans. Veh. Technol.}, vol.~67, no.~2,
  pp. 1830--1834, Feb 2018.

\bibitem{nikopour.2013}
H.~Nikopour and H.~Baligh, ``Sparse code multiple access,'' in \emph{IEEE 24th
  Annu. Int. Symp. on Personal, Indoor, and Mobile Radio Commun. (PIMRC)}, Sept
  2013, pp. 332--336.

\bibitem{taherzadeh.2014}
M.~Taherzadeh, H.~Nikopour, A.~Bayesteh, and H.~Baligh, ``{{SCMA}} codebook
  design,'' in \emph{IEEE 80th Veh. Technol. Conf.}, Sept 2014, pp. 1--5.

\bibitem{nikopour.2014}
H.~Nikopour, E.~Yi, A.~Bayesteh, K.~Au, M.~Hawryluck, H.~Baligh, and J.~Ma,
  ``{SCMA} for downlink multiple access of {5G} wireless networks,'' in
  \emph{IEEE Global Commun. Conf.}, Dec 2014, pp. 3940--3945.

\bibitem{xue.2016}
T.~Xue, L.~Qiu, and X.~Li, ``Resource allocation for massive {{M2M}}
  communications in {{SCMA}} network,'' in \emph{IEEE 84th Veh. Technol.
  Conf.}, Sept 2016, pp. 1--5.

\bibitem{luo.2017}
L.~Luo, L.~Li, and X.~Su, ``Optimization of resource allocation in relay
  assisted multi-user {{SCMA}} uplink network,'' in \emph{Int. Conf. on
  Computing, Networking and Commun.}, Jan 2017, pp. 282--286.

\bibitem{di.2016}
B.~Di, L.~Song, and Y.~Li, ``Radio resource allocation for uplink sparse code
  multiple access ({{SCMA}}) networks using matching game,'' in \emph{IEEE Int.
  Conf. on Commun.}, May 2016, pp. 1--6.

\bibitem{au.2014}
K.~Au, L.~Zhang, H.~Nikopour, E.~Yi, A.~Bayesteh, U.~Vilaipornsawai, J.~Ma, and
  P.~Zhu, ``Uplink contention based {SCMA} for {5G} radio access,'' in
  \emph{IEEE Globecom Workshops}, Dec 2014, pp. 900--905.

\bibitem{li.2016}
Z.~Li, W.~Chen, F.~Wei, F.~Wang, X.~Xu, and Y.~Chen, ``Joint codebook
  assignment and power allocation for {{SCMA}} based on capacity with gaussian
  input,'' in \emph{IEEE/CIC Int. Conf. on Commun. in China (ICCC)}, July 2016,
  pp. 1--6.

\bibitem{dabiri.2018}
M.~Dabiri and H.~Saeedi, ``Dynamic {SCMA} codebook assignment methods: A
  comparative study,'' \emph{{IEEE} Commun. Lett.}, vol.~22, no.~2, pp.
  364--367, Feb 2018.

\bibitem{cui.2017}
J.~Cui, P.~Fan, X.~Lei, Z.~Ma, and Z.~Ding, ``Downlink power allocation in
  {SCMA} with finite-alphabet constraints,'' in \emph{IEEE 85th Veh. Technol.
  Conf.}, June 2017, pp. 1--5.

\bibitem{liu.2017}
J.~Liu, M.~Sheng, L.~Liu, Y.~Shi, and J.~Li, ``Modeling and analysis of
  {{SCMA}} enhanced {{D2D}} and cellular hybrid network,'' \emph{{IEEE} Trans.
  Commun.}, vol.~65, no.~1, pp. 173--185, Jan 2017.

\bibitem{zhu.2017}
W.~Zhu, L.~Qiu, and Z.~Chen, ``Joint subcarrier assignment and power allocation
  in downlink {{SCMA}} systems,'' in \emph{{{IEEE}} 86th Veh. Tech. Conf.},
  Sept. 2017, pp. 1--5.

\bibitem{abedi.2018}
M.~R. Abedi, N.~Mokari, M.~R. Javan, and E.~A. Jorswieck, ``{Single or Multiple
  Frames Content Delivery for Next-Generation Networks?}'' \emph{ArXiv
  e-prints}, p. arXiv:1802.06910, Feb. 2018.

\bibitem{moltafet.2018b}
M.~{Moltafet}, P.~{Azmi}, N.~{Mokari}, M.~R. {Javan}, and A.~{Mokdad},
  ``Optimal and fair energy efficient resource allocation for energy
  harvesting-enabled-{PD-NOMA}-based hetnets,'' \emph{IEEE Transactions on
  Wireless Communications}, vol.~17, no.~3, pp. 2054--2067, March 2018.

\bibitem{mokdad.2019}
A.~{Mokdad}, P.~{Azmi}, N.~{Mokari}, M.~{Moltafet}, and M.~{Ghaffari-Miab},
  ``Cross-layer energy efficient resource allocation in {PD-NOMA} based
  {H-CRANs}: Implementation via {GPU},'' \emph{IEEE Transactions on Mobile
  Computing}, vol.~18, no.~6, pp. 1246--1259, June 2019.

\bibitem{mceliece.1998}
R.~J. McEliece, D.~J.~C. MacKay, and J.-F. Cheng, ``Turbo decoding as an
  instance of {P}earl's 'belief propagation' algorithm,'' \emph{{IEEE} J. Sel.
  Areas Commun.}, vol.~16, no.~2, pp. 140--152, Feb 1998.

\bibitem{3gpp}
``{Discussion on the feasibility of advanced {MU}-detector},'' {3GPP}, Tech.
  Rep. {TSG} {RAN} {WG1} Meeting 86, Aug 2016.

\bibitem{zou.2015}
J.~Zou, H.~Zhao, and W.~Zhao, ``Low-complexity interference cancellation
  receiver for sparse code multiple access,'' in \emph{IEEE 6th Int. Symp. on
  Microwave, Antenna, Propagation, and EMC Technologies}, Oct 2015, pp.
  277--282.

\bibitem{wei.2017}
F.~Wei and W.~Chen, ``Low complexity iterative receiver design for sparse code
  multiple access,'' \emph{{IEEE} Trans. Commun.}, vol.~65, no.~2, pp.
  621--634, Feb 2017.

\bibitem{cover}
T.~M. Cover and J.~A. Thomas, \emph{Elements of Information Theory}.\hskip 1em
  plus 0.5em minus 0.4em\relax Wiley-Interscience, 2006.

\bibitem{le.2018}
M.~T.~P. {Le}, G.~C. {Ferrante}, T.~Q.~S. {Quek}, and M.~{Di Benedetto},
  ``Fundamental limits of low-density spreading noma with fading,''
  \emph{{IEEE} Trans. Wireless Commun.}, vol.~17, no.~7, pp. 4648--4659, July
  2018.

\bibitem{zaidel.2018}
B.~M. {Zaidel}, O.~{Shental}, and S.~S. {Shitz}, ``Sparse {NOMA}: A closed-form
  characterization,'' in \emph{2018 IEEE International Symposium on Information
  Theory (ISIT)}, June 2018, pp. 1106--1110.

\bibitem{shental.2017}
O.~{Shental}, B.~M. {Zaidel}, and S.~S. {Shitz}, ``Low-density code-domain
  noma: Better be regular,'' in \emph{2017 IEEE International Symposium on
  Information Theory (ISIT)}, June 2017, pp. 2628--2632.

\bibitem{boyd.2004}
S.~Boyd and L.~Vandenberghe, \emph{Convex Optimization}.\hskip 1em plus 0.5em
  minus 0.4em\relax New York, NY, USA: Cambridge University Press, 2004.

\bibitem{lan.2010}
T.~Lan, D.~Kao, M.~Chiang, and A.~Sabharwal, ``An axiomatic theory of fairness
  in network resource allocation,'' in \emph{Proc. {{IEEE INFOCOM}}}, Mar.
  2010, pp. 1--9.

\bibitem{cvxpy}
S.~Diamond and S.~Boyd, ``{CVXPY}: A {P}ython-embedded modeling language for
  convex optimization,'' \emph{Journal of Machine Learning Research}, vol.~17,
  no.~83, pp. 1--5, 2016.

\bibitem{cvxpy_rewriting}
A.~Agrawal, R.~Verschueren, S.~Diamond, and S.~Boyd, ``A rewriting system for
  convex optimization problems,'' \emph{Journal of Control and Decision},
  vol.~5, no.~1, pp. 42--60, 2018.

\bibitem{domahidi.2013}
A.~Domahidi, E.~Chu, and S.~Boyd, ``{ECOS}: {A}n {SOCP} solver for embedded
  systems,'' in \emph{European Control Conference (ECC)}, 2013, pp. 3071--3076.

\bibitem{akle.2015}
S.~A. Serrano, ``Algorithms for unsymmetric cone optimization and an
  implementation for problems with the exponential cone,'' Ph.D. dissertation,
  Stanford University, Institute for Computational and Mathematical
  Engineering, 2015.

\bibitem{Patzold2012MobileChannels}
M.~{Patzold}, \emph{{Mobile Radio Channels}}.\hskip 1em plus 0.5em minus
  0.4em\relax Wiley Publishing, 2012.

\bibitem{nocedal.2006}
J.~Nocedal and S.~J. Wright, \emph{Numerical Optimization}, 2nd~ed.\hskip 1em
  plus 0.5em minus 0.4em\relax New York, NY, USA: Springer, 2006.

\bibitem{nesterov.1994}
{Nesterov, Y. and Nemirovskii, A.}, \emph{Interior-Point Polynomial Algorithms
  in Convex Programming}.\hskip 1em plus 0.5em minus 0.4em\relax Society for
  Industrial and Applied Mathematics, 1994.

\bibitem{hong.2017}
M.~Hong, X.~Wang, M.~Razaviyayn, and Z.-Q. Luo, ``Iteration complexity analysis
  of block coordinate descent methods,'' \emph{Mathematical Programming}, vol.
  163, no.~1, pp. 85--114, May 2017.

\bibitem{borndofer.2015}
R.~Bornd\"{o}rfer and O.~Heismann, ``The hypergraph assignment problem,''
  \emph{Discret. Optim.}, vol.~15, no.~C, pp. 15--25, Feb. 2015.

\end{thebibliography}
}
\end{document}